 \documentclass[10pt]{IEEEtran}
\usepackage[a4paper]{geometry}
\IEEEoverridecommandlockouts

\usepackage{cite}
\usepackage{amsmath,amssymb,amsfonts}
\usepackage{graphicx}
\usepackage{textcomp}
\usepackage{xcolor}
\usepackage{float}
\usepackage{array}
\usepackage{algorithm}
\usepackage{algpseudocode}
\usepackage{bm}
\usepackage{tabularx}
\usepackage{multirow}
\usepackage{booktabs}
\usepackage{stfloats}
\usepackage{ntheorem}

\usepackage{geometry}
\usepackage{subfigure}
\usepackage{stfloats}
\makeatletter
\renewcommand{\fnum@algorithm}{}

\usepackage{epsfig}
\usepackage{epstopdf}

\begin{document}
\title{Basis Expansion Extrapolation based Long-Term Channel Prediction for Massive MIMO OTFS Systems }

\author{\normalsize {Yanfeng~Zhang,~\IEEEmembership{\normalsize{Member,~IEEE}},}
        Xu~Zhu,~\IEEEmembership{\normalsize{Senior Member,~IEEE}},
        Yujie~Liu,~\IEEEmembership{\normalsize{Member,~IEEE}},\\
        Yong Liang~Guan,~\IEEEmembership{\normalsize{Senior Member,~IEEE}},
        David González~G.,~\IEEEmembership{\normalsize{Senior Member,~IEEE,}\vspace{-0pt}}
        and Vincent~K.~N.~Lau,~\IEEEmembership{\normalsize{Fellow,~IEEE}\vspace{-10pt}}
\vspace{-15pt}
\thanks{An earlier version of this paper was presented in part at IEEE ICC \mbox{2023 \cite{ZhangICC2023}.} (Corresponding author: Xu Zhu.)}
\thanks{Y. Zhang is with the School of Electrical Engineering and Intelligentization, Dongguan University of Technology, Dongguan, China. (e-mail: yfzhang@ieee.org)}
\thanks{X. Zhu is with the School of Electronic and Information Engineering, Harbin Institute of Technology, Shenzhen, China. (e-mail:xuzhu@ieee.org)}
\thanks{Y. Zhang, Y. Liu and Y. Guan are with Continental-NTU Corporate Lab, Nanyang Technological University, Singapore. (e-mail: yujie.liu@ieee.org, eylguan@ntu.edu.sg.)}
\thanks{D. González G. is with the Wireless Communications Technologies Group, Continental AG, Germany. (e-mail: david.gonzalez.g@ieee.org.)}
\thanks{Vincent K. N. Lau is with the Department of Electronic and Computer Engineering, The Hong Kong University of Science and Technology, Hong Kong, China. (e-mail:eeknlau@ust.hk.)}
\vspace{-10pt}}

\maketitle

\begin{abstract}
Massive multi-input multi-output (MIMO) combined with orthogonal time frequency space (OTFS) modulation has emerged as a promising technique for high-mobility scenarios. However, its performance could be severely degraded due to channel aging caused by user mobility and high processing latency. 
In this paper, an integrated scheme of uplink (UL) channel estimation and downlink (DL) channel prediction is proposed to alleviate channel aging in time division duplex (TDD) massive MIMO-OTFS systems.
Specifically, first, an iterative basis expansion model (BEM) based UL channel estimation scheme is proposed to accurately estimate UL channels with the aid of carefully designed OTFS frame pattern. Then a set of Slepian sequences are used to model the estimated UL channels, and the dynamic Slepian coefficients are fitted by a set of orthogonal polynomials. A channel predictor is derived to predict DL channels by iteratively extrapolating the Slepian coefficients. Simulation results verify that the proposed UL channel estimation and DL channel prediction schemes outperform the existing schemes in terms of normalized mean square error of channel estimation/prediction and DL spectral efficiency, with less pilot overhead. 
\end{abstract}


\vspace{-12pt}
\section{Introduction}
Orthogonal time frequency space (OTFS) has been proposed as a promising modulation technique for the sixth-generation wireless communications, due to its remarkable performance in high-mobility communications \cite{Wei2021}. OTFS with massive multi-input multi-output (MIMO) is able to further increase the spectral efficiency (SE) and system throughput \cite{Rama2018}. In addition, time division duplex (TDD) has been considered as a popular duplexing scheme thanks to its high SE, flexible adjustment to uplink (UL) and \mbox{downlink (DL)} traffics, channel reciprocity, etc \cite{Kim2020,Caojie2023,Dongzhihao2023}. 
Achieving high performance gains in the TDD massive MIMO-OTFS system is contingent upon accurately acquiring UL and DL channel state information (CSI).
However, channel aging arises in high-mobility communications, where both UL and DL channels change fast over time \cite{ZhangWCL,AIEXTA2023,Zheng2021}. Channel reciprocity is no longer suitable for high-mobility TDD massive MIMO-OTFS systems, which makes their channel estimation and prediction more challenging. Therefore, investigating integrated solutions of UL channel estimation and DL channel prediction for massive MIMO-OTFS systems becomes imperative.
 
 \vspace{-10pt}
\subsection{Related Work}
In static or low-speed scenarios, the wireless channel is assumed to be only affected by frequency selectivity fading, and the channel impulse response (CIR) remains unchanged within the symbol duration. However, in high-mobility scenarios, the channel additionally experiences time selective fading caused by high Doppler spread, thereby the resultant time-frequency (TF) doubly selective channel makes the CIR to be time-varying within the symbol duration \cite{Wei2021}. In order to obtain real-time CSI in high-mobility scenarios, a long training sequence is usually required for channel estimation to track the large number of channel parameters \cite{ZhangWCL}, which not only leads to high computational complexity, but also results in spectral inefficiency. Hence, new modulation schemes that are robust to channel time-variations are being extensively explored. To cope with this problem, OTFS modulation technology was proposed \cite{Hadani2017} and attracted much attention due to significant advantages in time-variant channels.

Channel estimation is one of the important research problems in OTFS to achieve high reliability communications in high-mobility scenarios. In recent years, a number of pilot-aided OTFS channel estimation algorithms have been proposed \cite{Raviteja2019, Shi2021,Sup2021,Sup2022,Shen2019, ULDL2020, Bayes2022,Message2021,Qu2021, Liu2022,Zhang2023cross}. In \cite{Raviteja2019} and \cite{Shi2021}, embedded pilot-aided channel estimation methods have been proposed, where a large number of guard pilots are deployed in the delay-Doppler (DD) domain to avoid strong inter-Doppler interference (IDI). Superimposed pilot-aided channel estimation has been reported in \cite{Sup2021} and \cite{Sup2022} to improve SE, where the interference between data and pilot is eliminated iteratively to ensure accurate channel estimation. Additionally, compressed sensing (CS) based channel estimation algorithms have been developed in \cite{Shen2019, ULDL2020,Message2021,Bayes2022} by exploiting the sparsity of channels in DD domain. However, only integer or fractional Doppler shifted channels are considered in these works. Recently, more practical channels with continuous Doppler spread have been studied in \cite{Qu2021,Liu2022,Zhang2023cross}, and basis expansion model (BEM) based channel estimation algorithms have been proposed to track time-varying channel. However, in these works, obtaining both UL and DL OTFS channels requires transmitting pilots in each OTFS frame, leading to a significant reduction in SE. DL channel prediction is an effective tool to reduce pilot overhead, whereby previous/historical UL channel estimation are used to predict DL channels without the need for pilot transmission. However, to the best of the authors' knowledge, OTFS DL channel prediction has not been extensively investigated in the literature.

A variety of channel prediction techniques have been proposed for non-OTFS systems, which can be broadly classified into three categories: parametric model-based \mbox{methods \cite{Para2009, Para2015,Para2018},} linear model-based methods \cite{Wu2021,Kalman2000,Talaei2021,Prony2020,MatrixPencil2022,Peng2019}, and neural network-based methods \cite{Dong2019,Jiang2022,Chu2022,Mattu2022}. Parametric model-based methods predict the physical parameters of wireless channels (\emph{e.g.}, complex attenuation, path delay, Doppler frequency, etc.) rather than directly predicting the CIRs \cite{Para2015}. Their prediction performance depends heavily on the estimation accuracy of physical parameters. In practical communication scenarios, these parameters are difficult to be accurately estimated due to dynamical channel environment \cite{Para2018}. 
Linear model-based prediction method treats the channel as a linear combination of several past channel coefficients and selects an appropriate optimization criterion to derive a channel predictor. Commonly used linear predictors include autoregressive (AR) \cite{Wu2021}, Kalman filter \cite{Kalman2000}, subspace \mbox{extraplotation \cite{Talaei2021},}  Prony's method \cite{Prony2020}, matrix pencil method \cite{MatrixPencil2022},  first-order Taylor model \cite{Peng2019}, etc. These methods perform well in slowly varying channels. However, this assumption may not be applicable in high-mobility scenarios.
Neural network-based channel prediction methods aim to forecast future channels by treating the prediction problem as a time-series prediction problem. A nonlinear prediction model can be learned from a large number of historical channel samples by using technologies such as deep convolutional neural network (CNN) \cite{Dong2019}, transformer \cite{Jiang2022}, reinforcement learning \cite{Chu2022} and long short-term memory (LSTM) network \cite{Mattu2022}.
These methods have demonstrated better prediction performance than parametric and linear methods in non-stationary environments. However, they are computationally expensive due to large-scale data training, and their hyperparameters are often set through trial and error, making real-time channel prediction and model transfer difficult. Moreover, the high-order nonlinear nature of neural networks makes it difficult to perform asymptotic analysis of channel prediction \mbox{performance.}

It is noteworthy that existing studies primarily focuses on either channel estimation or channel prediction, with little attention given to the integration of UL channel estimation and DL channel prediction in high-mobility scenarios.
In addition, a large number of guard pilots are required to avoid IDI and inter-user interference (IUI) in existing channel estimation schemes \cite{Shen2019, ULDL2020, Bayes2022,Message2021,Qu2021, Liu2022,Zhang2023cross}, resulting in low SE. Most importantly, due to high mobility, it is unreasonable to assume block fading channel in OTFS systems. As a result, existing channel prediction methods may not perform well in predicting a large number of future channel samples, such as those in multiple continuous OTFS frames. This motivates us to develop a unified scheme that integrates low-pilot-overhead UL channel estimation and long-term DL channel prediction for massive MIMO-OTFS systems.

 \vspace{-10pt}
\subsection{Contributions}
In this paper, an integrated scheme of UL channel estimation and DL channel prediction is proposed for massive MIMO-OTFS systems.  At the UL stage, an iterative BEM based UL channel estimation scheme is proposed. A variable block length-simultaneous orthogonal matching pursuit (VBL-SOMP) algorithm is then proposed to estimate the BEM coefficients by taking advantage of the individual and common block sparsity of channels. Then the Savitzky-Golay (SG) smoothing technique is designed to refine the channel estimation. Finally, a data-aided version is extended to further enhance UL channel estimation. At the DL stage, a set of Slepian sequences are used to represent the estimated UL channels, and a smoothed basis expansion \mbox{extrapolation (SBEE)} based channel prediction algorithm is developed to predict the DL Slepian coefficients instead of DL CIRs. The predicted DL Slepian coefficients are then used to recover the DL channels. The main contributions of this paper are summarized as follows.

\begin{itemize}
\item To the best of the authors' knowledge, this is the first solution that takes into account both UL channel estimation and DL channel prediction for TDD massive MIMO-OTFS systems. Unlike the existing schemes \cite{Raviteja2019, Shi2021,Sup2021,Sup2022,Shen2019, ULDL2020, Bayes2022,Message2021,Qu2021, Liu2022,Zhang2023cross} requiring pilot-aided estimation for both UL and DL channels, the proposed scheme needs to estimate UL channel only and DL channel can be predicted from UL channel estimates, enabling the reduction of pilot overhead and processing latency.
Furthuer more, since the proposed scheme only requires a small number of estimated UL channel samples, it is more suitable for practical communication systems. \mbox{In contrast,} the schemes \mbox{in \cite{Para2009,Para2015,Para2018,Wu2021,Kalman2000,Talaei2021,Prony2020,
MatrixPencil2022,Peng2019,Dong2019,Jiang2022,Chu2022,Mattu2022}} require a large number of perfect historical UL channel samples for DL \mbox{channel prediction.}

\item The proposed iterative BEM channel estimator enables accurate UL channel estimation with very low pilot overhead and fast convergence speed. Rather than using non-overlapping or orthogonal pilots in \cite{ULDL2020,Bayes2022,Message2021}, the pilots of proposed OTFS frame design can be overlapping for different users, thus greatly reducing pilot overhead. In addition, thanks to the exploitation of both individual and common block channel sparsity in BEM domain, the proposed VBL-SOMP algorithm outperforms the existing methods, \emph{e.g.}, \cite{Sup2021}, \cite{Shen2019} and \cite{Liu2022}, in terms of channel estimation.
Further more, the proposed SG smoothing method can reduce the channel modeling error introduced by BEM, thus speeding up the convergence of iterative channel estimation.

\item The proposed SBEE channel predictor can enhance both short-term and long-term channel prediction by exploiting the global temporal correlation of UL channel estimates and subsequent DL channel predicts. Besides, iterative extrapolation is performed for a small number of Slepian coefficients instead of predicting a large number of channel coefficients in \cite{Para2009,Para2015,Para2018,Wu2021,Kalman2000,Talaei2021,Prony2020,
MatrixPencil2022,Peng2019,Dong2019,Jiang2022,Chu2022,Mattu2022}. Hence, high-dimensional matrix operations and intensive offline training can be avoided, reducing computational complexity and training cost.
The asymptotic performance of the proposed SBEE channel predictor is also provided. Simulation results verify that the DL channel prediction performance of SBEE predictor outperforms the existing schemes \cite{Wu2021}, \cite{Prony2020} and \cite{Mattu2022} in terms of normalized mean square error (NMSE) of channel prediction and DL SE. 
\end{itemize}

 \vspace{-10pt}
\subsection{Organization and Notations}
\emph{Organization of the Paper:}
The system model and channel model for massive MIMO-OTFS systems are presented in Section II. In Section III, the iterative BEM based UL sparse channel estimation is described in detail. The proposed DL SBEE channel prediction scheme with UL channel estimation is highlighted in Section IV. Complexity and performance analysis of the proposed UL iterative BEM channel estimator and DL SBEE channel predictor are given in Section V. Extensive simulation results are shown in Section VI, and conclusions are drawn in Section VII.

\emph{Notations:} Bold symbols represent vectors or matrices. $ \otimes $ and $\odot $ denote the Kronecker product and Hadamard product, respectively. ${( \cdot )^{\rm T}}$, ${(\cdot)}^{*}$, ${( \cdot )^{\rm H}}$ and ${( \cdot )^{ - 1}}$ denote the transpose, complex conjugate, conjugate transpose and matrix inversion, respectively. $|| \cdot |{|_p}$ and $|| \cdot |{|_{\rm F}}$ denote ${\ell _p}$ norm operation and Frobenius norm operation, respectively. ${\jmath}=\sqrt{-1}$.

\vspace{-5pt}
\section{System Model}
In this paper, a massive MIMO-OTFS system is considered, where a single base station (BS) \footnote{It is noteworthy that the proposed integrated UL channel estimation and DL channel prediction scheme can be extended to multi-cell scenarios, which is left for future work. The inter-cell interference can be avoided by using existing techniques such as inter-BS interference coordination \cite{Hamza2013} and user scheduling scheme \cite{Xie2017}.} equipped with $N_{\rm r}$ antennas serves $N_{\rm u}$ users, each with a single antenna, as illustrated in Fig. 1. 
Considering the TDD mode, different time resources are allocated for UL and DL transmissions. As show in Fig, 2, at UL stage, the users send OTFS modulated signal that contain both data and pilots to BS. The BS extracts pilots for channel estimation, and then predicts DL channels based on the estimated UL channels. At DL stage, the BS performs precoding based on the predicted DL channels and transmits DL data to the users.

\vspace{-10pt}
\subsection{Signal Model of Massive MIMO-OTFS Systems}
Let ${\bf{x}}_{{\rm{DD}}}^{({n_{\rm u}})} = {[X_{{\rm{DD}}}^{({n_{\rm u}})}[0,0], \cdots ,X_{{\rm{DD}}}^{({n_{\rm u}})}[M - 1,N - 1]]^{\rm H}}\in {{\mathbb{C}}^{MN}}$ denote the DD-domain signal sent by the $n_{\rm u}$-th user, where $X_{\text{DD}}^{({{n}_{\text{u}}})}[m,n]$ is the $(m, n)$-th symbol in the DD-domain, $M$ and $N$ are the number of  delay bins and Doppler bins, respectively. Both transmit and receive pulse
shaping are assumed to be rectangular \cite{Rama2018}. The received DD-domain signal can be \mbox{written as}
\begin{equation}
{{\bf{y}}_{{\rm{DD}}}} = \sum\nolimits_{{n_{\rm u}} = 1}^{{N_{\rm u}}} {{\bf{H}}_{{\rm{eff}}}^{({n_{\rm u}})}{\bf{x}}_{{\rm{DD}}}^{({n_{\rm u}})} + {{\bf{w}}_{{\rm{DD}}}}} ,
\label{eq1}
\end{equation}
where ${{\bf{y}}_{{\rm{DD}}}} = {[{({\bf{y}}_{{\rm{DD}}}^{(1)})^{\rm H}}, \cdots ,{({\bf{y}}_{{\rm{DD}}}^{({N_{\rm r}})})^{\rm H}}]^{\rm H}}\in {{\mathbb{C}}^{MN{{N}_{\text{r}}}}}$, ${\bf{y}}_{{\rm{DD}}}^{({n_{\rm r}})} = {[y_{{\rm{DD}}}^{({n_{\rm r}})}[0], \cdots ,y_{{\rm{DD}}}^{({n_{\rm r}})}[MN - 1]]^{\rm H}}$ is the received signal at the $n_{\rm r}$-th antenna, ${\bf{H}}_{{\rm{eff}}}^{{(n_{\rm u})}} = {[{({\bf{H}}_{{\rm{eff}}}^{(1, n_{\rm u})})^{\rm H}}, \cdots ,{({\bf{H}}_{{\rm{eff}}}^{({N_{\rm r}, n_{\rm u}})})^{\rm H}}]^{\rm H}}\in {{\mathbb{C}}^{{{N}_{\text{r}}}MN\times MN}}$, $\mathbf{H}_{\text{eff}}^{({{n}_{\text{r}}},{{n}_{\text{u}}})}\in {{\mathbb{C}}^{MN\times MN}}$ is the equivalent channel matrix between the $n_{\rm u}$-th user and the $n_{\rm r}$-th BS antenna, and ${{\bf{w}}_{{\rm{DD}}}} = {[{({\bf{w}}_{{\rm{DD}}}^{(1)})^{\rm H}}, \cdots ,{({\bf{w}}_{{\rm{DD}}}^{({N_{\rm r}})})^{\rm H}}]^{\rm H}}\in {{\mathbb{C}}^{MN{{N}_{\text{r}}}}}$. Specifically, ${{{\bf{w}}_{\rm{DD}}^{({n_{\rm r}})}}}$ is the additive white Gaussian noise (AWGN) and ${{\bf{w}}_{\rm DD}^{({n_{\rm r}})}} \sim {\cal C}{\cal N}(0,{\sigma ^2}{{\bf{I}}_{MN}})$, where ${{\mathbf{I}}_{MN}}\in {{\mathbb{R}}^{MN\times MN}}$ is an identity matrix, and $\sigma ^2$ is the average power of AWGN. ${\bf{H}}_{{\rm{eff}}}^{({n_{\rm r}},{n_{\rm u}})}$ can be expressed as \cite{Rama2018}
\begin{equation}
{\bf{H}}_{{\rm{eff}}}^{({n_{\rm r}},{n_{\rm u}})} = {{\bf{B}}_{\rm{r}}}{\bf{H}}_{\rm{T}}^{({n_{\rm r}},{n_{\rm u}})}{{\bf{B}}_{\rm{t}}},
\label{eq2}
\end{equation}

\begin{figure}[htbp]
\centerline{\includegraphics[width=0.5\textwidth]{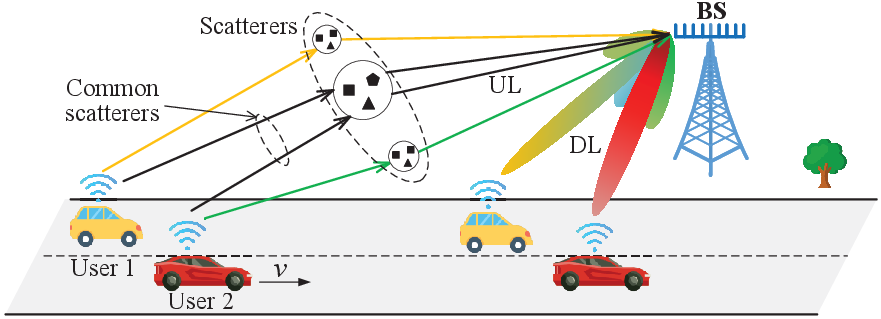}}
\caption{High-mobility massive MIMO-OTFS systems with two users.}
\label{fig1}
\vspace{-5pt}
\end{figure}
\begin{figure}[htbp]
\centerline{\includegraphics[width=0.48\textwidth]{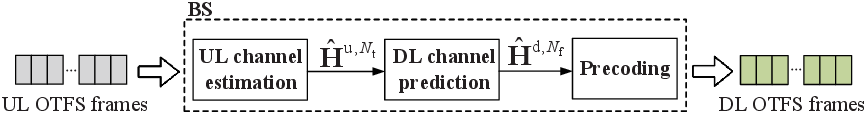}}
\caption{Block diagram of the integrated scheme of UL channel estimation and DL channel prediction.}
\vspace{-10pt}
\label{fig2}
\end{figure}

\noindent where ${\bf{H}}_{\rm{T}}^{({n_{\rm r}},{n_{\rm u}})}\in {{\mathbb{C}}^{MN\times MN}}$ is the time-domain channel matrix, ${{\bf{B}}_{\rm{r}}} = {{\bf{F}}_N} \otimes {{\bf{I}}_M}$, ${{\bf{B}}_{\rm{t}}} = {\bf{F}}_N^{\rm H} \otimes {{\bf{I}}_M}$, and ${{\bf{F}}_N} \in {\mathbb{C}^{N \times N}}$ is a normalized discrete Fourier transform (DFT) matrix.

\vspace{-10pt}
\subsection{Channel Model}
Considering a broadband channel with $L$ paths, the time-domain channel matrix corresponding to the $n_{\rm r}$-th antenna and $n_{\rm u}$-th user can be represented as
\begin{equation}
{\bf{H}}_{\rm{T}}^{({n_{\rm r}, n_{\rm u}})} = \sum\nolimits_{l = 0}^{L - 1} {{{\bf{\Pi }}^l}{\rm{diag}}\left( {{{\bf{h}}_{{n_{\rm r}, n_{\rm u}},l}}} \right)} ,
\label{eq3}
\end{equation}
where ${\bf{\Pi }} = {\rm{circ}}\{ [0,1,0, \cdots ]_{MN}^{\rm T}\} \in {{\mathbb{C}}^{MN\times MN}}$ is a permutation matrix, ${{\bf{h}}_{{n_{\rm r}},{n_{\rm u}},l}} = {[{h_{{n_{\rm r}},{n_{\rm u}},l,1}}, \cdots ,{h_{{n_{\rm r}},{n_{\rm u}},l,MN}}]^{\rm H}}$ is the CIR of  \mbox{$l$-th} path. 
The number of paths can be calculated as 
$L = \left\lfloor {{\tau _{\max }}/{T_{\rm s}}} \right\rfloor $, where $\tau_{\max}$ denotes the maximum delay and $T_{\rm s}$ is the sampling period. The Jakes' model is used to characterize the time-varying CIR ${{\bf{h}}_{{n_{\rm u}},{n_{\rm r}},l}}$, resulting in the correlation function of the channel of the $l$-th path being ${J_0}(2\pi n{f_{\max }}{T_{\rm s}})$, where ${J_0}( \cdot )$ represents the zero-order Bessel function of the first kind and ${f_{\max }} = \frac{{{f_{\rm c}}v}}{c}$ is the maximum Doppler frequency. 
The temporal correlation between the $k$-th channel gain at the $n_{\rm t}$-th frame and the $j$-th channel gain at the $n_{\rm f}$-th frame is given as $\mathbb{E}\left\{ {{h}_{({{n}_{\text{t}}}-1)MN+k}}{{[{{h}_{({{n}_{\text{f}}}-1)MN+j}}]}^{*}} \right\}={{J}_{0}}(2\pi {{f}_{\max }}{{T}_{s}}|({{n}_{\text{t}}}-{{n}_{\text{f}}})MN+k-j|)$, where $k, j =0,1,\cdots,MN-1$.
Parametric model is adopted to characterize the spatial correlation of channels in massive MIMO systems \cite{Para2018}, which is expressed as a function of channel gain and angle of arrival (AOA). Assuming that there are multiple scattering clusters between the users and BS, and each cluster contains multiple rays. Without loss of generality, the number of clusters is assumed to be equal to the number of paths, ${h_{{n_{\rm r}},{n_{\rm u}},l,n}}$ can then be written as
\begin{equation}
{h_{{n_{\rm r}},{n_{\rm u}},l,n}} = {\beta _{{n_{\rm u}},l,n}}{a_{{\rm{BS}}}}({\theta _{{n_{\rm r}},{n_{\rm u}},l}}),
\label{eq4}
\end{equation}
where ${\beta _{{n_{\rm u}},l,n}} = {\zeta _{{n_{\rm u}},l,n}}{e^{ - \jmath (2\pi {f_{\max }}n{T_{\rm s}}\cos {\psi _{{n_{\rm u}},l}} + {{\bar \psi }_{{n_{\rm u}},l}})}}$ with ${\zeta _{{n_{\rm u}},l,n}}$ being the $n_{\rm u}$-th user's time-varying complex amplitude of the $l$-th path at time instant $n$, ${{\psi _{n_{\rm u},l}}}$ denotes the horizontal angle between the $n_{\rm u}$-th user's motion direction and the BS, ${\bar{\psi} _{n_{\rm u},l}} \in [0,2\pi ]$ is the initial phase,  and ${a_{{\rm{BS}}}}({\theta _{{n_{\rm r}},{n_{\rm u}},l}})$ is the $n_{\rm r}$-th element of steering vector given by
\begin{equation}
{a_{{\rm{BS}}}}({\theta _{{n_{\rm r}},{n_{\rm u}},l}}) = {e^{j2\pi {n_{\rm r}}d\sin {\theta _{{n_{\rm u}},l}}/{{\lambda}_{w}} }},{n_{\rm r}} = 1, \cdots ,{N_{\rm r}},
\label{eq5}
\end{equation}
where ${\theta _{n_{\rm u},l}} $ is the AOA of the $l$-th path seen by the BS, $\lambda_{\rm w} $ is the wavelength of the transmitted signal, $d$ denotes the antenna spacing, which is typically set to $d = {{\lambda}_{\rm w}  \mathord{\left/
 {\vphantom {\lambda  2}} \right.
 \kern-\nulldelimiterspace} 2}$. In far-field communications, the incident signal seen by the BS is assumed to be limited to a narrow angular spread \cite{Para2018}, the AOA $\theta_{n_{\rm u},l}$ can be expressed as ${\theta _{n_{\rm u},l}} = {{\bar \theta }_{n_{\rm u},l}} + \Delta {\theta _{n_{\rm u},l}}$, where ${{\bar \theta }_{n_{\rm u},l}}$ and ${\Delta _{{\theta _{n_{\rm u},l}}}}$ denote the $n_{\rm u}$-th user's central angle and angular spread corresponding to the $l$-th cluster, respectively.
 
 In this paper, $N_{\rm t}$ consecutive UL OTFS subframes and $N_{\rm f}$ consecutive DL OTFS subframes can be regarded as one frame, and the variation of the channel gains within this frame is modeled by the Jakes’ Doppler spectrum \cite{Zhang2022}. According to \cite{Qinqibo2018}, AOA is assumed to remain unchanged within $N_{\rm f}$ consecutive OTFS frames if ${{N}_{\text{f}}}<\frac{\pi D}{2{{T}_{\text{f}}}v{{N}_{\text{r}}}}$, where $D$  denotes the distance between the scatterers and the mobile user, and $T_{\rm f}$ is the frame duration of OTFS systems. For example, provided with $D = 100$ meters, $T_{\rm f}=267$ us, $N_{\rm r} = 64$, and $v = 120$ km/h, AOA would not change with time as long as ${{N}_{\text{f}}}<276$. Therefore, AOA is able to maintain unchanged within tens of OTFS frames.
 
Considering outdoor macro-cell scenarios with rich scatterers \cite{Zhang2022}, the channel is assumed to be sparse in delay domain, and there are only $K$ ($K\ll L$) significant paths within the maximum delay, \emph{i.e.},
\begin{equation}
K = \sum\nolimits_{k = 0}^{L - 1} {{\rm{sgn}}(||{{\bf{h}}_{{n_{\rm r}},{n_{\rm u}},k}}|{|_2} - \zeta )},
\label{eq6}
\end{equation}
where ${\mathop{\rm sgn}} (x) = 1$ if $x\ge 0$, and ${\mathop{\rm sgn}} (x) = 0$ if $x< 0$. $\zeta $ is a small positive threshold, which determines whether the channel power of the $k$-th path can be ignored. It is noteworthy that the number of significant paths $K$ can be estimated by existing methods, such as minimum description length criterion \cite{MDL2009}. 
In practical communication systems, such as the internet of vehicles, physically adjacent users may share several common scatterers, as shown in \mbox{Fig. 1,} leading to a partial overlap of their channels in delay domain \cite{Rao2014}. This means that each user's channel consists of common paths and individual paths. In Section III, the common-path and individual-path characteristics of channels will be exploited to design low-overhead channel estimation algorithms for massive MIMO-OTFS systems.

As in \cite{Poutanen2010}, it is reasonable to assume that different users share a common path if there is a common scatterer between these users and the angular separation of the common scatterer seen from the users is less than 90 degrees. According to \cite{Zhang2022}, path delays are proved to vary more slowly than path gains. This is because the coherence time of path gains and the duration for path delay variation are inversely proportional to the carrier frequency and the signal bandwidth, while the carrier frequency is much larger than the signal bandwidth in a practical system. For example, given the mobile speed $v$ = 33.3 m/s (120 km/h), number of subcarriers $M=128$, subcarrier spacing $\Delta f = 30$ KHz, sampling period ${T_{\rm{s}}} = 1/M\Delta f = 2.6 \times {10^{ - 7}}$ seconds, and CP length $L_{\rm {CP}}= 128$ instants, the path delays can be assumed unchanged over $J \le \frac{{0.01c}}{{(M + {L_{{\rm{CP}}}})v}} = 351$ consecutive OFDM symbols \cite{Qinqibo2016}, corresponding to a duration of 0.0234 seconds, where $c$ is the speed of light. In addition, a duration of 0.0234 seconds is equivalent to 78 OTFS frames each containing $N = 8$ OFDM symbols. Therefore, it is assumed in this paper that the channel path delays remain unchanged for the duration of 78 consecutive OTFS symbols.


\vspace{-10pt}
\section{BEM Iterative UL Channel Estimation}
A BEM based iterative channel estimator is proposed in this section, as shown in Fig 3. A novel OTFS frame pattern with hybrid pilots is designed to reduce pilot overhead. A VBL-SOMP algorithm is then proposed to estimate the sparse UL channel by taking advantage of both the channel's common-block-sparse and individual-block-sparse properties. Additionally, a SG based channel smoothing approach is proposed to refine the estimated channel, reducing the channel modeling error caused by BEM. Furthermore, a data-aided iterative scheme is introduced to further improve the accuracy of UL channel estimation.

\begin{figure}
\centerline{\includegraphics[width=0.48\textwidth]{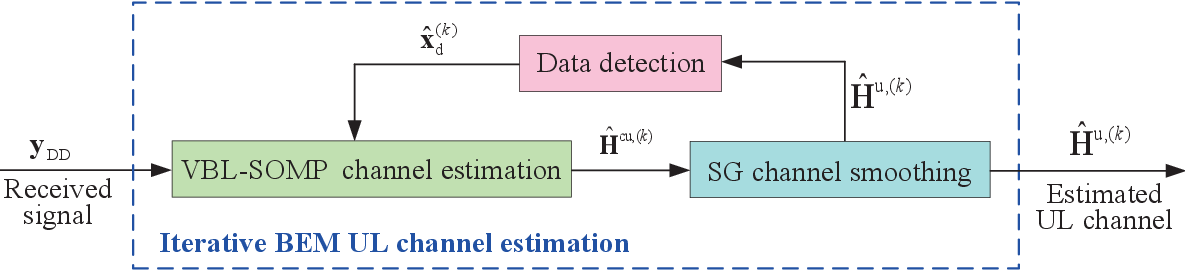}}
\caption{Block diagram of the iterative BEM UL channel estimation scheme.}
\vspace{-5pt}
\label{fig3}
\end{figure}

\begin{table}
\renewcommand{\arraystretch}{1.25}
\centering
\caption{Key Notations in Subsection III-A}
\begin{tabular}{|c|c|}
\hline
\textbf{Symbol} & \textbf{Definition} \\
\hline
$Q$ & Order of CE-BEM \\
\hline
$G$ & Number of non-zero pilots per user \\
\hline
${\bf{x}}_{\rm{p}}^{({n_{\rm u}})} \in \mathbb{C}^{MN}$ & TF-domain pilot of the $n_{\rm u}$-th user \\
\hline
${\bf{x}}_{{\rm{nz}}}^{({n_{\rm u}})}\in \mathbb{C}^{G}$ & Non-zero pilot of the $n_{\rm u}$-th user \\
\hline
${{\mathcal{P}}_{\text{nz}}}\in {{\mathbb{R}}^{G}}$ & Index set of non-zero pilots \\
\hline
${{\mathcal{P}}_{q}}\in {{\mathbb{R}}^{G}}(q\ne (Q-1)/2)$& Index set of guard pilots \\
\hline
${\bf{x}}_{{\rm{DD}},{\rm{P}}}^{({n_{\rm u}})}\in {\mathbb{C}^{MN \times MN}}$& DD-domain pilot of the $n_{\rm u}$-th user \\
\hline
${\lambda _{\rm{p}}} = \frac{{G(2Q - 1)}}{{MN}}$& Pilot overhead of the proposed scheme \\
\hline
\end{tabular}
\label{tab2}
\end{table}

\vspace{-10pt}
\subsection{OTFS Frame Pattern}
A novel OTFS frame with hybrid pilots is designed to enable sparse channel estimation for massive MIMO-OTFS systems. As shown in Fig. 4, for the $n_{\rm u}$-th user, the original TF-domain pilot signal ${\bf{x}}_{\rm{p}}^{({n_{\rm u}})} \in \mathbb{C}^{MN}$ is carefully designed with a number of non-zero pilots (in orange color), each of which is surrounded by $2Q-2$ guard pilots (in white color), where $Q$ is the CE-BEM order. The non-zero pilots are denoted as ${\bf{x}}_{{\rm{nz}}}^{({n_{\rm u}})}\in \mathbb{C}^{G}$, where $G$ is the number of non-zero pilots. The optimization of $G$ non-zero pilots is given in Subsection III-C. Define $\{ {{\cal P}_q}\} _{q = 0}^{Q - 1}$ as the pilot index set of ${\bf{x}}_{\rm{p}}^{({n_{\rm u}})}$. ${{\cal P}_{\rm nz}}$ and ${{\cal P}_q}$ ($q \ne \frac{{Q - 1}}{2}$) are the corresponding index sets of non-zero pilots and their surrounding guard pilots. The original pilot signal ${\bf{x}}_{\rm{p}}^{({n_{\rm u}})}$ is then transformed into the DD domain by an unitary transformation:
\begin{equation}
{\bf{x}}_{{\rm{DD}},{\rm{P}}}^{({n_{\rm u}})} = ({{\bf{F}}_N} \otimes {{\bf{I}}_M}){\bf{F}}_{MN}^{\rm H}{\bf{x}}_{\rm{p}}^{({n_{\rm u}})} = {\bf{Px}}_{\rm{p}}^{({n_{\rm u}})},
\label{eq7}
\end{equation}
where ${\bf{x}}_{{\rm{DD}},{\rm{P}}}^{({n_{\rm u}})}\in {\mathbb{C}^{MN \times MN}}$ is the DD-domain pilots, ${\bf{P}} = ({{\bf{F}}_N} \otimes {{\bf{I}}_M}){{\bf{F}}^{\rm H}}\in {{\mathbb{C}}^{MN\times MN}}$ is a unitary matrix and ${\bf{F}} \in {\mathbb{C}^{MN \times MN}}$ is a DFT matrix. Key symbols in this subsection are notated in Table I.

Note that after unitary transformation the number of pilots increases from $G(2Q-1)$ to $MN$. To preserve the low overhead of the original pilot signal, $G(2Q-1)$ resource grids, whose indices are chosen from $\cal P$, are used for pilot transmission only. Such pilots are called dedicated pilots. The remaining resources grids are then all used for data transmission. Note that $MN-G(2Q-1)$ pilots would be transmitted along with data symbols using the same resources (in gray color), and they are called superimposed pilots. The dedicated
pilot overhead of the proposed frame could be calculated as ${\lambda _{\rm{p}}} = \frac{{G(2Q - 1)}}{{MN}}$, which is much lower than the pilot overhead ${\lambda _{\rm{p}}} = \frac{{(2L-1)(2Q - 1)}}{{MN}}$ of the existing scheme in \cite{Liu2022} because $2L-1$ is typically much larger than \mbox{$G$ \cite{Zhang2022}.}
It is noteworthy that the guard pilots are necessary in our channel estimation scheme due to the use of BEM based channel modeling and CS based channel estimation. According to the pilot pattern shown in Fig. 4, the receiver can obtain $Q$ interference-free received pilot signals, which enables us to leverage the structured sparsity of the channel for designing high-accuracy channel estimation algorithm. When some pilots or data were deployed in the guard region, the dedicated pilots would suffer from the interferences from non-zero pilots and data symbols, resulting in a mixed signal that cannot be separated into $Q$ interference-free, and thus leading to performance degradation.

\begin{figure*}
\centerline{\includegraphics[width=0.75\textwidth]{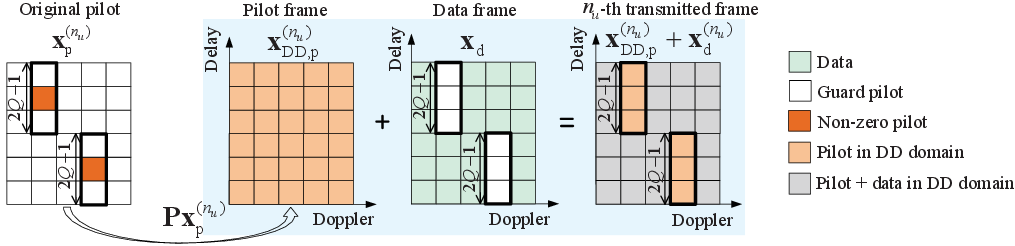}}
\caption{Proposed OTFS frame pattern with hybrid pilots for the $n_{\rm u}$-th user. All users share the same pilot positions in DD domain.}
\vspace{-15pt}
\label{fig4}
\end{figure*}

As shown in Fig. 4, $N_{\rm u}$ users share a common original pilot positions ${\cal P}_{\rm nz}$ and ${\cal P}$. To facilitate sparse channel estimation, each user is assigned with a random $ \pm 1$ sequence as non-zero pilots. Unlike the conventional OTFS frame pattern \cite{Shen2019} that requires non-overlapping  DD domain resources for different users, the proposed OTFS frame pattern allows all users to use the same pilot region. Thus, the number of users supported by the proposed pilot pattern is not limited by the resolution of DD grids. This allows a larger number of users to communicate with the BS with very low pilot overhead.

\vspace{-5pt}
\subsection{VBL-SOMP based UL Channel Estimation}
In high-mobility massive MIMO systems, wireless channels have both strong spatial and time correlation \cite{Para2018}. In this subsection, the spatial-time correlation of channels is exploited to reduce the complexity of channel estimation. Specifically, the complex exponential BEM (CE-BEM) \cite{Liu2022} and the spatial rotation BEM (SR-BEM) \cite{Xie2017} are used to model the channel in time domain and spatial domain, respectively. Thus the high-dimensional channel estimation problem can be converted into a low-dimensional BEM coefficient estimation problem. In time domain, the time-varying UL CIRs can be modeled as
\begin{equation}
{\bf{h}}_{{n_{\rm r}},{n_{\rm u}},l}^{{\rm{UL}}} = \sum\nolimits_{q = 0}^{Q - 1} {{{\bf{b}}_q}{c_{{n_{\rm r}},{n_{\rm u}},l,q}} + {{\bf{v}}_{{n_{\rm r}},{n_{\rm u}},l}}},
\label{eq8}
\end{equation}
where ${\bf{h}}_{{n_{\rm r}},{n_{\rm u}},l}^{{\rm{UL}}}\in {{\mathbb{C}}^{MN}}$ is the UL CIR vector, ${{\bf{b}}_q} = [B(0,q), \cdots ,$ $B(MN - 1,q){]^{\rm H}}$ is the $q$-th complex exponential basis function with ${{B}(n,q)} = {e^{\jmath {\omega _q}n}}$ with ${\omega _q} = \frac{{2\pi }}{{MN}}(q - \frac{{Q - 1}}{2})$, $Q$ is the BEM order and typically selected as $Q \ge 2\left\lceil {N{f_{\max }}{T_{\rm s}}} \right\rceil  + 1$, ${{c_{{n_{\rm r}},{n_{\rm u}},l,q}}}$ is the CE-BEM coefficient and ${{\bf{v}}_{{n_{\rm r}},{n_{\rm u}},l}}$ denotes channel modeling error.

In \cite{Para2018} and \cite{Xie2017}, the DFT matrix has been used to model channel in angle domain.
However, in practical communication systems, the angular resolution of antenna is limited by the number of antennas. The columns of the DFT matrix may not match the steering vectors, resulting in power leakage. To alleviate power leakage\footnote{Note that existing methods such as energy-focusing windowing \cite{Window2019}, dictionary learning \cite{Diclearning2018}, off-grid method \cite{Offgrid2018}, etc., can also be used to alleviate power leakage issues.}, the SR-BEM \cite{Xie2017} is adopted to model the CE-BEM coefficients
\begin{equation}
{{{\bf{ c}}}_{{n_{\rm u}},l,q}} = \left( {{\bf{r}}({\vartheta _{{n_{\rm u},l}}}) \odot {\bf{F}}_{{N_{\rm r}}}^{\rm H}} \right){{\bf{s}}_{{n_{\rm u}},l,q}} + {\bf{w}}_{{\rm mod},{\rm{s}}}^{({n_{\rm u}},l,q)},
\label{eq9}
\end{equation}
where ${{{\bf c}}_{{n_{\rm u}},l,q}} = {[{{c}}_{1,{n_{\rm u}},l,q}, \cdots ,{{c}}_{{N_{\rm r}},{n_{\rm u}},l,q}]^{\rm H}}$, ${\bf{r}}({\vartheta _{{n_{\rm u},l}}}) \odot {\bf{F}}_{{N_{\rm r}}}^{\rm H}$ is the basis matrix of the SR-BEM, ${{\bf{F}}_{{N_{\rm r}}}} \in {\mathbb{C}^{{N_{\rm r}} \times {N_{\rm r}}}}$ is a DFT matrix, ${{\bf{s}}_{{n_{\rm u}},l,q}}$ is SR-BEM coefficients, ${\bf{r}}({\vartheta _{{n_{\rm u}}}}) = {[1,{e^{j{\vartheta _{{n_{\rm u},l}}}}}, \cdots ,{e^{j({N_{\rm r}} - 1){\vartheta _{{n_{\rm u},l}}}}}]^{\rm H}}$ denotes the rotation vector with an rotation angle $\vartheta_{n_{\rm u},l}  \in [ - {\pi  \mathord{\left/
 {\vphantom {\pi  {{N_{\rm r}},}}} \right.
 \kern-\nulldelimiterspace} {{N_{\rm r}},}}{\pi  \mathord{\left/
 {\vphantom {\pi  {{N_{\rm r}}}}} \right.
 \kern-\nulldelimiterspace} {{N_{\rm r}}}}]$ and ${\bf{w}}_{{\rm mod},{\rm{s}}}^{({n_{\rm u}},l,q)}$ is the modeling error. The optimal rotation angle can be obtained by solving the optimization problem
 \begin{equation}
\vartheta _{{n_{\rm u}}}^* = \mathop {\arg \min }\limits_{\vartheta  \in [ - \frac{\pi }{{{N_{\rm r}}}},\frac{\pi }{{{N_{\rm r}}}}]} \sum\limits_{l = 0}^{L - 1} {\sum\limits_{q = 0}^{Q - 1} {||{{{\bf{\tilde c}}}_{{n_{\rm u}},l,q}} - ({\bf{r}}(\vartheta ) \odot {\bf{F}}_{{N_{\rm r}}}^{\rm H}){{{\bf{\tilde s}}}_{{n_{\rm u}},l,q}}|{|^2}} },
 \label{eq10}
 \end{equation}
where ${{{\bf{\tilde c}}}_{{n_{\rm u}},l,q}} = B(n,q){{{\bf{\tilde h}}}_{{n_{\rm u}},l,n}}$, ${{{\bf{\tilde h}}}_{{n_{\rm u}},l,n}} \in {\mathbb{C}^{{N_{\rm r}} \times 1}}$ denotes the coarse channel estimate obtained using the preamble \cite{Xie2017}, and ${{{\bf{\tilde s}}}_{{n_{\rm u}},l,q}} = {({\bf{r}}({\vartheta _{{n_{\rm u}},l}}) \odot {\bf{F}}_{{N_{\rm r}}}^{\rm H})^{ - 1}}{{{\bf{\tilde c}}}_{{n_{\rm u}},l,q}}$.
The optimal ${Q_{\rm{s}}}$ columns are selected to construct a ${Q_{\rm{s}}}$-order SR-BEM of the $n_{\rm u}$-th user according to the indices of the largest ${Q_{\rm{s}}}$ elements in ${{{\bf{\tilde s}}}_{{n_{\rm u}},l,q}}$, \emph{i.e.}, ${{\bf{D}}_{n_{\rm u}}} = {[{\bf{r}}({\vartheta _{n_{\rm u}}^ * }) \odot {\bf{F}}_{{N_{\rm r}}}^{\rm H}]_{\cal D}}$, where ${\cal D}$ is the index set of the ${Q_{\rm{s}}}$  largest elements in ${{{\bf{\tilde s}}}_{{n_{\rm u}},l,q}}$.
Then the CE-BEM coefficients can be modeled by the SR-BEM as
\begin{equation}
{c_{{n_{\rm r}},{n_{\rm u}},l,q}} = \sum\nolimits_{{q_{\rm s}} = 0}^{{Q_{\rm{s}}} - 1} {{D_{{n_{\rm u}}}}[{n_{\rm r}},{q_{\rm s}}]{s_{{n_{\rm u}},l,q,{q_{\rm s}}}}}  + {{{\bf{\bar v}}}_{{n_{\rm u}},l,q}},
\label{eq11}
\end{equation}
where ${{D_{{n_{\rm u}}}}[{n_{\rm r}},{q_{\rm s}}]}$ is the $(n_{\rm r},q_{\rm s})$-th element of ${\bf{D}}_{n_{\rm u}}$, ${{s_{{n_{\rm u}},l,q,{q_{\rm s}}}}}$ is the SR-BEM coefficient, and ${{{\bf{\bar v}}}_{{n_{\rm u}},l,q}}$ is the modeling error introduced by SR-BEM. Substituting (11) and (8) into (3), the time-domain channel can be further expressed as a function of the SR-BEM coefficients as
\begin{equation}
\begin{array}{l}
{\bf{H}}_{\rm{T}}^{({n_{\rm r}},{n_{\rm u}})} = \sum\limits_{l = 0}^{L - 1} {\sum\limits_{q = 0}^{Q - 1} {{{\bf{\Pi }}^l}{\rm{diag}}} \left( {{{\bf{b}}_q}} \right)\left( {{{\bf{D}}_{{n_{\rm u}}}}({n_{\rm r}},:) \otimes {{\bf{I}}_L}} \right){{\bf{s}}_{{n_{\rm u}},q}}} \\
 + {\bf{W}}_{{\rm{mod}},{\rm{T}}}^{({n_{\rm r}},{n_{\rm u}})} + {\bf{W}}_{{\rm{mod}},{\rm{S}}}^{({n_{\rm r}},{n_{\rm u}})}
\end{array},
 \label{eq12}
\end{equation}
where ${{{\bf{D}}_{n_{\rm u}}}({n_{\rm r}},:)}\in {{\mathbb{C}}^{1\times {{Q}_{\text{s}}}}}$ denotes the $n_{\rm r}$ row of ${\bf{D}}_{n_{\rm u}}$, ${\bf{W}}_{\bmod ,{\rm{T}}}^{({n_{\rm r}},{n_{\rm u}})}\in {{\mathbb{C}}^{MN\times MN}}$ and ${\bf{W}}_{\bmod ,{\rm{S}}}^{({n_{\rm r}},{n_{\rm u}})}\in {{\mathbb{C}}^{MN\times MN}}$ are the modeling errors of CE-BEM and SR-BEM, respectively.  Substituting (12) into (2) and (1), the received signal can be 
decomposed into $Q$ separate equations as
\begin{equation}
{\left[ {{\bf{\bar U}}_q^{\rm H}} \right]_{{{\cal P}_q}}}{{\bf{y}}_{{\rm{DD}}}} = {\left[ {\bf{\Psi }} \right]_{{{\cal P}_{\frac{{Q - 1}}{2}}}}}{\bf{\bar D}}{{\bf{s}}_q} + {{{\bf{\tilde z}}}_q},{\rm{  }}q = 0, \cdots ,Q - 1,
 \label{eq13}
\end{equation}
where ${\left[ {{{{\bf{\bar U}}}_q}} \right]_{{{\cal P}_q}}} = {{\bf{I}}_{{N_{\rm r}}}} \otimes {[{{\bf{U}}_q}]_{{{\cal P}_q}}}\in {{\mathbb{C}}^{G{{N}_{\text{r}}}\times G{{N}_{\text{r}}}}}$ and ${{\bf{U}}_q} = {{\bf{B}}_{\rm{r}}}{\mathop{\rm diag}\nolimits} \left( {{{\bf{b}}_q}} \right){\bf{F}}_{MN}^{\rm H}\in {{\mathbb{C}}^{MN\times MN}}$ is a unitary matrix, ${\bf{\Psi }} = [{{\bf{\Psi }}^{(1)}}, \cdots ,{{\bf{\Psi }}^{({N_{\rm u}})}}]\in {{\mathbb{C}}^{MN{{N}_{\text{r}}}\times MNL{{N}_{\text{u}}}}}$ and ${{{\bf{\Psi }}^{({n_{\rm u}})}}} = ({{\bf{I}}_{{N_{\rm r}}}} \otimes {{\mathop{\rm diag}\nolimits} ({\bf{x}}_{\rm{p}}^{({n_{\rm u}})})})({{\bf{I}}_{{{MN}}}} \otimes {{\bf{F}}_L})\in {{\mathbb{C}}^{MN{{N}_{\text{r}}}\times MNL}}$ is an equivalent pilot matrix, ${\bf{\bar D}} = {\mathop{\rm blkdiag}\nolimits} \{ {{\bf{D}}_1} \otimes {{\bf{I}}_L}, \cdots ,{{\bf{D}}_{{N_{\rm u}}}} \otimes {{\bf{I}}_L}\} \in {{\mathbb{C}}^{L{{N}_{\text{r}}}\times L{{Q}_{\text{s}}}{{N}_{\text{u}}}}}$, ${{\bf{s}}_q} = {[{s_{1,0,q,0}}, \cdots ,{s_{{n_{\rm u}},l,q,{q_{\rm s}}}}, \cdots ,{s_{{N_{\rm u}},L - 1,q,{Q_{\rm{S}}} - 1}}]^{\rm H}}\in {{\mathbb{C}}^{L{{Q}_{\text{s}}}{{N}_{\text{u}}}}}$, and ${{{\bf{\tilde z}}}_q}$ is an error term including AWGN and the modeling error of CE-BEM and SR-BEM. Please refer to the Appendix A for the derivation of (13).
By combining the $Q$ equations in (13), a compact form can be obtained as
\begin{equation}
{\bf{Y}} = {\bf{\Phi S}} + {\bf{\tilde Z}},
 \label{eq14}
\end{equation}
where ${\bf{Y}} = [{[{{{\bf{\tilde y}}}_{{\rm{DD}}}}]_{{{\cal P}_0}}}, \cdots ,{[{{{\bf{\tilde y}}}_{{\rm{DD}}}}]_{{{\cal P}_{Q - 1}}}}] \in {{\mathbb{C}}^{G{N_{\rm{r}}} \times Q}}$, ${{[{{{\mathbf{\tilde{y}}}}_{\text{DD}}}]}_{{{\mathcal{P}}_{q}}}}$$={{\left[ \mathbf{\bar{U}}_{q}^{\text{H}} \right]}_{{{\mathcal{P}}_{q}}}}{{\mathbf{y}}_{\text{DD}}}$, ${\bf{\Phi }} = {{\bf{\Psi }}_{{{\cal P}_{\frac{{Q - 1}}{2}}}}}\bar {\bf{D}} \in {{\mathbb{C}}^{G{{N}_{\text{r}}}\times L{{Q}_{\text{s}}}{{N}_{\text{u}}}}}$, ${\bf{S}} = [{{\bf{s}}_0}, \cdots ,{{\bf{s}}_{Q - 1}}]\in {{\mathbb{C}}^{L{{Q}_{\text{s}}}{{N}_{\text{u}}}\times Q}}$ and ${\bf{\tilde Z}} = [{{{\bf{\tilde z}}}_0}, \cdots ,{{{\bf{\tilde z}}}_{Q - 1}}]$.

Note that after CE-BEM and SR-BEM modeling, only  ${N_{\rm u}}{Q_{\rm{}}}{Q_{\rm{s}}}L$ unknown BEM coefficients in (14), instead of $MN{N_{\rm r}}{N_{\rm u}}{L}$ unknown channel parameters in ${{\bf{H}}_{\rm{T}}}$, need to be estimated, which greatly reduces the complexity of  UL channel estimation. Moreover, since there are only $K$ significant paths among $L$ multipaths, there are only $N_{\rm u}{Q_{\rm S}}K$ non-zero elements in $ {{\bf{s}}_q} $, $q=0,\cdots,Q-1$, which inspires us to adopt CS technique to solve (14).

To exploit the structured sparsity, we rearrange the elements of $\bf S$ in a block-wise manner, (14) can be \mbox{rewritten as}
\begin{equation}
{\bf{Y}} = \underbrace {{\bf{\Phi \Theta }}}_{{\bf{\bar \Phi }}}\underbrace {{{\bf{\Theta }}^{\rm T}}{\bf{S}}}_{{\bf{\bar S}}} + {\bf{\tilde Z}},
\label{eq15}
\end{equation}
where ${\bf{\Theta }} = [{{\bf{\Theta }}_1},{{\bf{\Theta }}_2}, \cdots ,{{\bf{\Theta }}_{L}}] \in {\mathbb{C}^{{N_{\rm u}}L{Q_{\rm{S}}} \times {N_{\rm u}}L{Q_{\rm{S}}}}}$ is a permutation matrix with ${{\bf{\Theta }}_k} = [{{\bf{e}}_k},{{\bf{e}}_{L + k}}, \cdots ,{{\bf{e}}_{({Q_{\rm{S}}{N_{\rm u}}} - 1)L + k}}]$, ${{\bf{e}}_k}$ is the $k$-th column of the identity matrix ${{\bf{I}}_{{Q_{\rm{S}}{N_{\rm u}}}L}}$. 
Consequently, $\bf{\bar S}$ shows a block-sparse structure with varying block lengths, as illustrated in Fig. 5.
Without loss of generality, it is assumed that the number of common paths of $N_{\rm u}$ users is $K_{\rm C}$, and thus each user has $K-K_{\rm C}$ individual paths. Clearly, the block sparsity level of ${\bf {\bar S}}$ is given as
\begin{equation}
\bar K = \left\{ {\begin{array}{*{20}{l}}
{K,\quad {{\Gamma}_1} =  \cdots  = {{\Gamma}_{{N_{\rm u}}}}}\\
{K{N_{\rm u}},\quad{{\Gamma}_1} \ne  \cdots  \ne {{\Gamma}_{{N_{\rm u}}}}}\\
{{N_{\rm u}}(K - {K_{\rm C}}) + {K_{\rm C}},\quad{\rm{others}}}
\end{array}} \right.,
\label{eq16}
\end{equation}
where ${\Gamma}_{n_{\rm u}}$ is the support of SR-BEM coefficients corresponding to the $n_{\rm u}$-th user.
Define the total common block support and individual block support of $\bf{\bar S}$ as ${\cal C}(\Gamma )$ and ${\cal I}(\Gamma )$, respectively,
where $\Gamma  = {\Gamma _1} \cap  \cdots  \cap {\Gamma _{{N_{\rm u}}}}$ is the row support of $\bf{\bar S}$. Then $\bf{\bar S}$ can be recovered by solving an optimization problem as follows
\begin{equation}
\begin{array}{l}
\mathop {\min }\limits_{{\bf{\bar S}}} ||{\bf{Y}} - {\bf{\bar \Phi }}{{{\bf{\bar S}}}_\Gamma }||_{\rm f}^2\\
{\rm{s}}{\rm{.t}}{\rm{.  }}\quad{\cal C}(\Gamma ) = {K_{\rm C}},{\cal I}(\Gamma ) = {N_{\rm u}}(K - {K_{\rm C}}).
\end{array}
\label{eq17}
\end{equation}


In this paper, the channel sparsity parameters $K$ and $K_{\rm C}$ are assumed to available to the BS. The channel sparsity $K$, can be obtained before implementing the proposed channel estimator by the priori-information aided iterative hard threshold algorithm \cite{Gaozhen2015}, adaptive support-aware algorithm \cite{Maxu2017TVT}, sparsity adaptive matching pursuit algorithm \cite{ZhangYi2016} etc. For example, the sparsity $K$ can be estimated by using the pseudo-random noise (PN) sequences. The transmitter sends a PN sequence with length $L_{\rm{PN}}$  ($L_{\rm{PN}} > L_{\rm{CP}}$). Then the receiver performs correlation calculation on the received PN sequences and the local PN sequences \cite{Maxu2017TVT}. Finally, the sparsity $K$ can be estimated by counting the number of correlation values greater than a predefined power threshold. As indicated in \cite{Rao2014}, the number of common paths $K_{\rm C}$ can be acquired from offline channel propagation measurement at the BS or long-term stochastic learning and estimation. For instance, the common sparsity $K_{\rm C}$ can be estimated by measuring the angles and distances of propagation scatterer centers of different users \cite{Rao2014}.

The multiple measurement vector problem in (17) can be solved by using existing CS algorithms such as simultaneous orthogonal matching pursuit (SOMP) algorithm \cite{Shen2019} and block SOMP (BSOMP) algorithm \cite{Gong2017}. However, the non-zero block length of $\bf{\bar S}$ is not fixed, which will degrade the recovery performance of SOMP and BSOMP algorithms. 

\begin{figure}
\centerline{\includegraphics[width=0.49\textwidth]{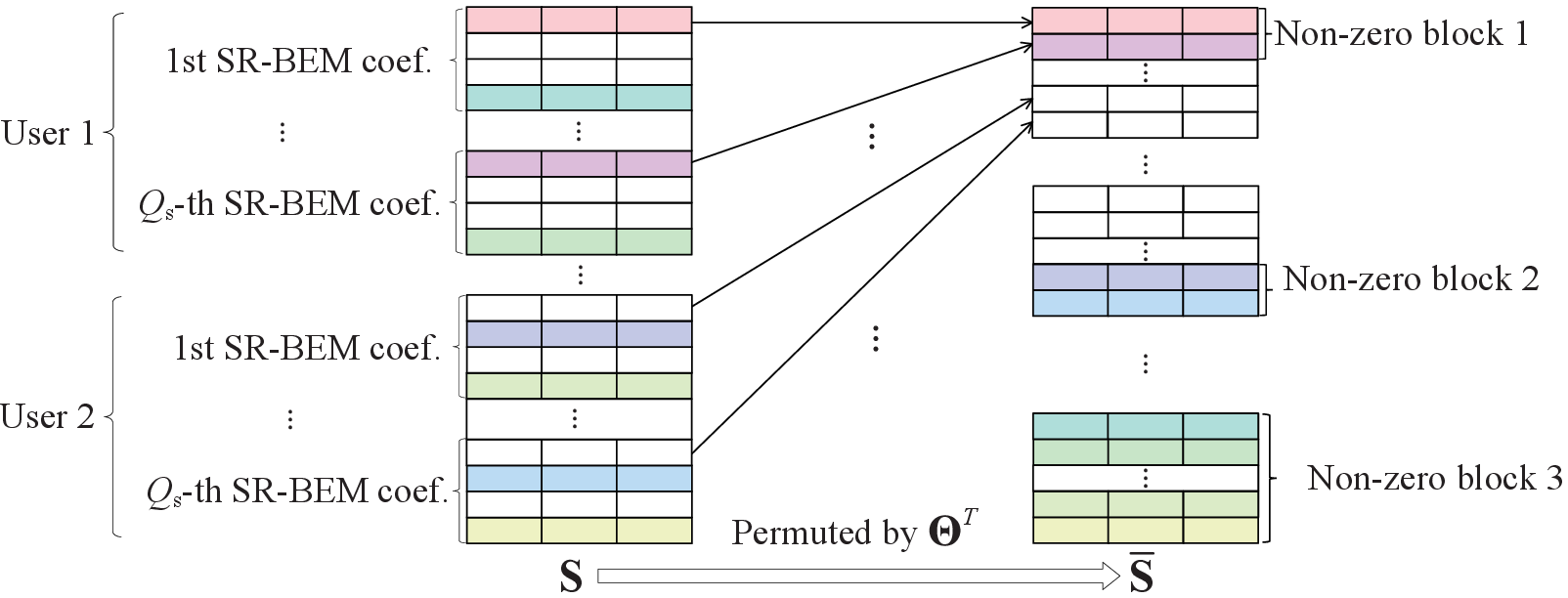}}
\caption{The block sparse structure of ${{\bf{\bar S}}}$ with $L = 4$, $N_{\rm u}=2$ and $Q_{\rm s}=2$. (The 1st and 4th paths of user 1 are non-zero, and the 2nd and 4th paths of user 2 are non-zero.)\vspace{-5pt}}
\vspace{-15pt}
\label{fig5}
\end{figure}

Inspired by the common block and individual block sparse structure of $\bf{\bar S}$, an variable-block-length SOMP (VBL-SOMP) algorithm is proposed to solve (17). The operation procedure of the VBL-SOMP algorithm is shown in {\bf{Algorithm 1}}.
{\bf{Algorithm 1}} mainly consists of two modules, steps (1)-(5) is designed to solve the common block sparse solution of SR-BEM coefficients, while the steps (6)-(9) is designed to solve its individual block sparse solution. It is noteworthy that the proposed VBL-SOMP algorithm has higher recovery accuracy than the traditional SOMP \cite{Shen2019} and BSOMP \cite{Gong2017} algorithms. The main reason lies in two aspects. On the one hand, the VBL-SOMP algorithm deals with common sparse blocks and individual sparse blocks separately, and its total block sparsity is smaller than the BSOMP algorithm that treats all blocks as individual sparse blocks. In the case of the same number of pilots, the smaller the block sparsity, the higher the recovery accuracy of channel estimation. On the other hand, the VBL-SOMP algorithm makes full use of the grouping characteristics of sparse blocks. In each iteration, VBL-SOMP first selects the most correlated group individual sparse blocks, and then identifies individual sparse blocks within the group, which enhances the identification accuracy of the support of individual sparse blocks. After ${{\bf{\hat{ \bar S}}}}$ is recoverd, we can obtain the coarse UL channel estimates ${\bf{\hat h}}_{{n_{\rm r}},{n_{\rm u}},l}^{{\rm{CUL}}}$ by substituting the estimated SR-BEM coefficient matrix ${\bf{\hat S}} = {\bf{\Theta \hat {\bar S}}}$ into (11) and (8).

\vspace{-0pt}
\subsection{ Savitzky-Golay Channel Smoothing}
In this paper, the BEM is used to model the doubly selective channels, and the UL channel is obtained by estimating the BEM coefficients. However, the BEM modeling error leads misalignment between the true and the estimated UL channels in the header and tail regions of multiple consecutive OTFS frames, as shown by the green line in Fig. 6. This will not only reduce the accuracy of UL channel estimation, but also degrades DL channel prediction accuracy. To reduce the UL channel estimation error, the SG smoothing method \cite{Schafer2011} is used to smooth the estimated UL channel samples. The key motivation for choosing the SG smoothing method lies in its ability to enhance the accuracy of UL channel estimation while preserving the channel’s shape characteristics (variation trend) \cite{Schafer2011}. The smoothed UL channel is then used for DL channel prediction.
SG is a data smoothing method based on local least squares (LS) polynomial approximation. Its basic idea is to construct a $Q_{\rm {sg}}$-order polynomial to fit a set of $2N_{\rm{sg}}+1$ samples centered on the current sample, and then use the obtained polynomial coefficients to reevaluate the current sample. In this paper, we propose a SG method to refine the estimated UL channel in a parallel manner. The estimated channels corresponding to $N_{\rm t}$ UL frames can be stacked as ${{{\bf{\hat H}}}^{{\rm{CUL}}}} = {[{({\bf{\hat H}}_1^{{\rm{CUL}}})^{\rm H}}, \cdots ,{({\bf{\hat H}}_{{N_{\rm t}}}^{{\rm{CUL}}})^{\rm H}}]^{\rm H}} \in {\mathbb{C}^{MN{N_{\rm t}} \times {N_{\rm r}}{N_{\rm u}}L}}$, where ${\bf{\hat H}}_{{n_{\rm t}}}^{{\rm{CUL}}} = [{\bf{\hat h}}_{0,1,0}^{{\rm{CUL,}}{n_{\rm t}}}, \cdots ,{\bf{\hat h}}_{{N_{\rm r}-1},{N_{\rm u}},L}^{{\rm{CUL,}}{n_{\rm t}}}]$ is the estimated UL channel matrix of the $n_{\rm t}$-th frame. First, for the $n_{\rm {sg}}$ -th channel sample, a smoothing window range ${N_{{\rm{sg}}}} \le {n_{{\rm{sg}}}} \le {N_{{\rm{sg}}}}$ is selected. Then, these $2{N_{{\rm{sg}}}} + 1$ UL channel samples are used to fit a polynomial with order ${Q_{{\rm{sg}}}}$, \emph{i.e.}, $h(x) = {c_0}{x^0} + {c_1}{x^1} +  \cdots  + {c_{{Q_{{\rm{sg}}}}}}{x^{{Q_{{\rm{sg}}}} - 1}}$, where $\{ {x^q}\} _{q = 0}^{{Q_{{\rm{sg}}}} - 1}$ is the polynomial function. According to the polynomial function, a basis matrix can be constructed as
\begin{equation}
\footnotesize{{{\bf{B}}_{{\rm{SG}}}} = \left[ {\begin{array}{*{20}{c}}
1&{ - {N_{{\rm{sg}}}}}&{{{( - {N_{{\rm{sg}}}})}^2}}& \cdots &{{{( - {N_{{\rm{sg}}}})}^{{Q_{{\rm{sg}}}}}}}\\
1&{ - {N_{{\rm{sg}}}} + 1}&{{{( - {N_{{\rm{sg}}}} + 1)}^2}}& \cdots &{{{( - {N_{{\rm{sg}}}} + 1)}^{{Q_{{\rm{sg}}}}}}}\\
 \vdots & \vdots & \vdots & \vdots & \vdots \\
1&0&{{0^2}}& \cdots &{{0^{{Q_{{\rm{sg}}}}}}}\\
 \vdots & \vdots & \vdots & \vdots & \vdots \\
1&{{N_{{\rm{sg}}}}}&{N_{{\rm{sg}}}^2}& \cdots &{N_{{\rm{sg}}}^{{Q_{{\rm{sg}}}}}}
\end{array}} \right].}
\label{eq18}
\end{equation}

The principle of SG channel smoothing is to reevaluate the $n_{\rm{sg}}$-th channel coefficient based on the fitted polynomial coefficients.
To smooth the $k$-th row of ${{{\bf{\hat H}}}^{{\rm{CUL}}}}$, the $(2{N_{\rm{sg}}} + 1)$ rows surrounding the $k$-th row of ${{{\bf{\hat H}}}^{{\rm{CUL}}}}$ are extracted as the input matrix, \emph{i.e.}, ${{{\bf{\hat H}}}^{{\rm{CUL,}}k}} = {[{({{{\bf{\hat H}}}^{{\rm{CUL}}}}(k - {N_{\rm{sg}}},:))^{\rm H}}, \cdots ,{({{{\bf{\hat H}}}^{{\rm{CUL}}}}(k + {N_{\rm{sg}}},:))^{\rm H}}]^{\rm H}}$. Then the polynomial coefficient matrix can be obtained by
\begin{equation}
{\bf{\hat C}}_{{\rm{SG}}}^k = {({\bf{B}}_{{\rm{SG}}}^{\rm H}{{\bf{B}}_{{\rm{SG}}}})^{ - 1}}{\bf{B}}_{{\rm{SG}}}^{\rm H}{{{\bf{\hat H}}}^{{\rm{CUL,}}k}},
\label{eq19}
\end{equation}

\renewcommand{\algorithmicrequire}{\textbf{Input:}}
\renewcommand{\algorithmicensure}{\textbf{Output:}}
\begin{algorithm}
\footnotesize
\caption{\footnotesize {\bf {Algorithm 1~}} VBL-SOMP for UL Channel Estimation} 
  \begin{algorithmic}[1]
    \Require Received signals: ${\bf{Y}}$. Measurement matrix: ${\bf{\bar \Phi }}$. Number of non-zero channel paths $K$ and number of overlapped non-zero paths $K_{\rm C}$.
       \State Initialize the iteration counter $k=0$,  
the block support ${\bf{\Gamma }}^{(0)} = {[{\bf{\Gamma }}_0^{\rm T}, \cdots ,{\bf{\Gamma }}_{L - 1}^{\rm T}]^{\rm T}} = \emptyset $ and the residual ${{\bf{R}}^{(0)}} = {\bf{Y}}$.
       \For {$k\le{N_{\rm u}}(K - {K_{\rm C}}) + {K_{\rm C}}$}
      \State {Calculate block correlation ${\beta _l} = \sum\nolimits_i {\sum\nolimits_j {|{{\Lambda } _l}(i,j)|} } $ with \mbox{${{{\bf{\Lambda }}}_l} = {\bf{\bar \Phi }}_l^{\rm H}{{\bf{R}}^{(k)}}$ for $l \in [0,L - 1]$.}}
      \State Find the index $m$ of the maximum value in $\{ {\beta _l}\} _{l = 0}^{L - 1}$
       \If {$k \le K_{\rm C}$}
       \State Update support set ${\bf{\Gamma }}^{(k)} \leftarrow \{ {\bf{\Gamma }}:{{\bf{\Gamma }}_m} = {{\bf{1}}_{{N_{\rm u}}{Q_{\rm{S}}} \times 1}}\} $.
      \State Update residual ${{\bf{R}}^{(k)}} \leftarrow {\bf{Y}} - {{{\bf{\bar \Phi }}}_{{{\bf{\Gamma }}_m}}}({\bf{\bar \Phi }}_{{{\bf{\Gamma }}_m}}^\dag {\bf{Y}})$.
      \State Update measurement matrix ${\bf{\bar \Phi }} \leftarrow \{ {\bf{\bar \Phi }}:{{{\bf{\bar \Phi }}}_{{{\bf{\Gamma }}_m}}} = {\bf{0}}\}  $.
      \Else 
      \State Calculate correlation ${\beta _{l,{n_{\rm u}}}} = \sum\nolimits_i {\sum\nolimits_j {|{{\Lambda } _{l,{n_{\rm u}}}}(i,j)|} } $  with  ${{\bf{{\Lambda } }}_{l,{n_{\rm u}}}} = {\bf{\bar \Phi }}_{l,{n_{\rm u}}}^{\rm H}{{\bf{R}}^{(k)}}$ for ${n_{\rm u}} \in [1,{N_{\rm u}}]$.
     \State  Find the index $(n,n_{\rm u})$ of the largest value in ${\bm{\beta }} = $
 $[{\beta _{0,1}}, \cdots ,{\beta _{L - 1,{N_{\rm u}}}}]$.
     \State Update support set ${\bf{\Gamma }}^{(k)} \leftarrow \{ {\bf{\Gamma }}:{{\bf{\Gamma }}_{n,n_{\rm u}}} = {{\bf{1}}_{{Q_{\rm{S}}} \times 1}}\} $.
     \State Update residual ${{\bf{R}}^{(k)}} = {\bf{Y}} - {{{\bf{\bar \Phi }}}_{{{\bf{\Gamma }}_{n,{n_{\rm u}}}}}}({\bf{\bar \Phi }}_{{{\bf{\Gamma }}_{n,{n_{\rm u}}}}}^\dag {\bf{Y}})$.
      \State Update ${\bf{\bar \Phi }} \leftarrow \{ {\bf{\bar \Phi }}:{{{\bf{\bar \Phi }}}_{{{\bf{\Gamma }}_{n,n_{\rm u}}}}} = {\bf{0}}\}  $.
      \EndIf
       \State $k=k+1$;
      \EndFor
    \Ensure The estimated non-zero SR-BEM coefficients ${{{\bf{\widehat {\bar S}}}}_{{\rm{nz}}}} = {\bf{\bar \Phi }}_{|{\bf{\Gamma }}^{(k)}|}^\dag {\bf{Y}}$ and the complete SR-BEM coefficient ${\bf{\hat {\bar S}}}({\bf{\Gamma }}^{(k)}) = {{{\bf{\hat {\bar S}}}}_{{\rm{nz}}}}$.
  \end{algorithmic}
\end{algorithm}
\begin{figure}
\centerline{\includegraphics[width=0.49\textwidth]{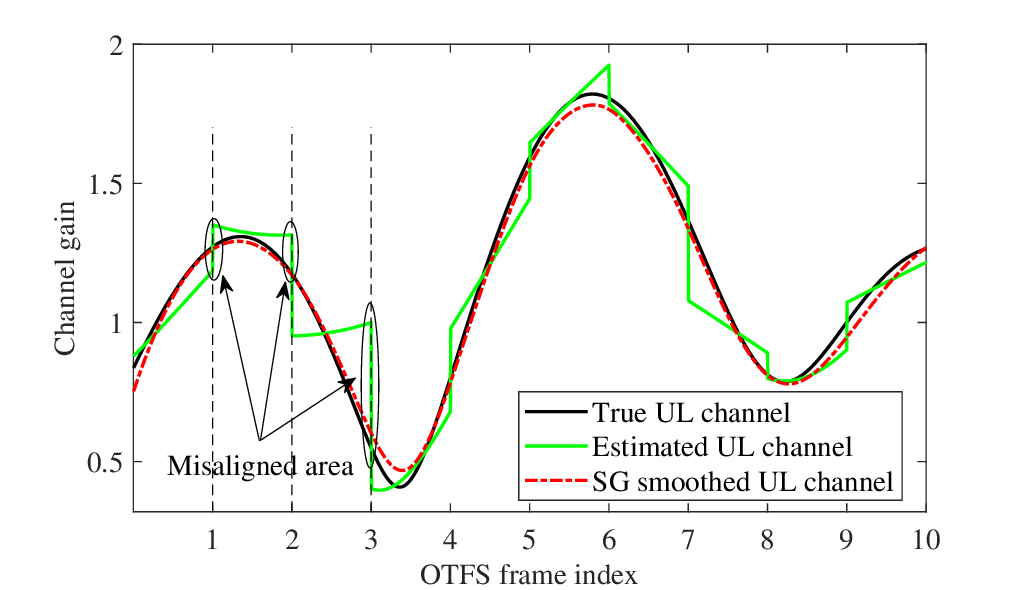}}
\caption{A snapshot of the estimated UL channel, the smoothed UL channel estimates and the true UL channel.\vspace{-0pt}}
\label{fig6}
\end{figure}

\noindent where ${\bf{\hat C}}_{{\rm{SG}}}^k$ is the fitted coefficient matrix. Then the smoothed $k$-th row of ${{{\bf{\hat H}}}^{{\rm{UL}}}}$ can be obtained as
\begin{equation}
{{{\bf{\hat H}}}^{{\rm{SG,}}k}} = {{\bf{B}}_{{\rm{SG}}}}(k,:){\bf{\hat C}}_{{\rm{SG}}}^k,
\label{eq20}
\end{equation}
where ${{{\bf{\hat H}}}^{{\rm{SG,}}k}}$ is the $k$-th row of ${{{\bf{\hat H}}}^{{\rm{UL}}}}$ smoothed by SG. We can sweep through $k=0,\cdots,MNN_{\rm t}-1$ to smooth all of ${{{\bf{\hat H}}}^{{\rm{CUL}}}}$.
Finally, the smoothed UL channel can be obtained as ${{{\bf{\hat H}}}^{{\rm{UL}}}} = {[{({{{\bf{\hat H}}}^{{\rm{SG}}}}(1,:))^{\rm H}}, \cdots ,{({{{\bf{\hat H}}}^{{\rm{SG}}}}(MN{N_{\rm t}},:))^{\rm H}}]^{\rm H}}$. Note that the polynomial order should satisfy ${Q_{{\rm{sg}}}} < 2{N_{{\rm{sg}}}} + 1$ to ensure the accuracy of solving for the polynomial coefficients, and the smoothing window length $2{N_{{\rm{sg}}}} + 1$ should be less than the number of channel samples $MN$. For example, in the simulation results, the smoothing window length and polynomial order are set to $2{N_{{\rm{sg}}}} + 1=11$ and ${Q_{{\rm{sg}}}}=5$, respectively. It can be seen from Fig. 6 that SG smoothing method can effectively address the misalignment problem, thereby significantly improving the accuracy of channel estimation.

\vspace{-0pt}
\subsection{Data Aided Channel Estimation Enhancement}
In order to further improve the accuracy of UL channel estimation, a data-aided channel estimation enhancement scheme is proposed by iteratively eliminating the interference between received pilot signals and data signals. At the BS, the received DD-domain data signal can be detected by using the LMMSE algorithm \cite{Ahmad2022}:
\begin{equation}
{{{\bf{\hat x}}}_{\rm{d}}} = {\left( {{\bf{\hat H}}_{{\rm{eff,d}}}^{\rm H}{{{\bf{\hat H}}}_{{\rm{eff,d}}}} + {\rm{SN}}{{\rm{R}}_{\rm{d}}}{{\bf{I}}_{{N_{\rm d}}}}} \right)^{ - 1}}{\bf{\hat H}}_{{\rm{eff,d}}}^{\rm H}{{\bf{y}}_{\rm{d}}},
\label{eq21}
\end{equation}
where ${{{\bf{\hat H}}}_{{\rm{eff,d}}}} = [{\bf{\hat H}}_{{\rm{eff,d}}}^{(1)}, \cdots ,{\bf{\hat H}}_{{\rm{eff,d}}}^{({N_{\rm u}})}]$, ${\bf{\hat H}}_{{\rm{eff,d}}}^{({n_{\rm u}})}\in {{\mathbb{C}}^{MN{{N}_{\text{r}}}\times {{N}_{\text{d}}}}}$ is the estimated DD-domain channel corresponding to the data symbols, $N_{\rm d}$ is the number of data symbols, ${{\mathop{\rm SNR}\nolimits} _{\rm{d}}} = {\sigma ^2}/\sigma _{\rm{d}}^2$ denotes the ratio of noise power to the average power of data, and ${\bf y}_{\rm d}$ is the received data signal. 
Note that the complexity of the LMMSE detector in (21) is $O(N_{\rm d}^3 + N_{\rm d}^2{M^2}{N^2}N_{\rm r}^2)$, which is unaffordable when the values of the number of data symbols is large. The complexity of the LMMSE algorithm in (21) mainly lies in the matrix inversion operations and matrix multiplication. Since ${\sigma ^2}/\sigma _{\rm{d}}^2 > 0$ for finite SNR ranges, the conjugate gradient algorithm \cite{Yin2014} can be used to reduce the complexity of the matrix inversion operation by equivalently formulating (21) as ${{\bf{r}}_{\rm{d}}} = {{\bf{\Xi }}}{{\bf{x}}_{\rm{d}}}$, where ${{\bf{r}}_{\rm{d}}} = {\bf{\hat H}}_{{\rm{eff,d}}}^{\rm H}{{\bf{y}}_{\rm{d}}}$ and ${{\bf{\Xi }}} = {\bf{\hat H}}_{{\rm{eff,d}}}^{\rm H}{{{\bf{\hat H}}}_{{\rm{eff,d}}}} + {{\mathop{\rm SNR}\nolimits} _{\rm{d}}}{{\bf{I}}_{{N_{\rm d}}}}$. The sparse structure of ${{{\bf{\hat H}}}_{{\rm{eff,d}}}}$ is exploited to further reduce the complexity of ${{\bf{\hat H}}_{{\rm{eff,d}}}^{\rm H}{{{\bf{\hat H}}}_{{\rm{eff,d}}}}}$. ${\bf{\Xi }}$ can be expressed as ${\bf{\Xi }} = {\mathop{\rm blkdiag}\nolimits} \{ {{\bf{\Xi }}_{1,1}}, \cdots ,{{\bf{\Xi }}_{{N_{\rm u}},{N_{\rm u}}}}\}  + {\rm{SN}}{{\rm{R}}_{\rm{d}}}{{\bf{I}}_{{N_{\rm d}}}} $, where ${{\bf{\Xi }}_{m,n}} = \sum\nolimits_{{n_{\rm r}} = 1}^{{N_{\rm r}}} {{{({\bf{H}}_{{\rm{eff,d}}}^{({n_{\rm r}},m)})}^{\rm H}}{\bf{H}}_{{\rm{eff,d}}}^{({n_{\rm r}},n)}}$ which can be calculated as
\begin{equation}
{{\bf{\Xi }}_{m,n}} = \sum\nolimits_{{n_{\rm r}} = 1}^{{N_{\rm r}}} {\sum\nolimits_{i = 1}^{KN} {\sum\nolimits_{j = 1}^{KN} {{{{\bf{\mathord{\buildrel{\lower3pt\hbox{$\scriptscriptstyle\frown$}} 
\over H} }}}_{{n_{\rm r}},m,i}}} } } {{{\bf{\mathord{\buildrel{\lower3pt\hbox{$\scriptscriptstyle\smile$}} 
\over H} }}}_{{n_{\rm r}},m,j}},
\label{eq22}
\end{equation}
where ${{{\bf{\mathord{\buildrel{\lower3pt\hbox{$\scriptscriptstyle\frown$}} 
\over H} }}}_{{n_{\rm r}},m,i}} = {{\bf{\Pi }}^{{{\cal U}_m}(i)}}{\mathop{\rm diag}\nolimits} ({{{\bf{\mathord{\buildrel{\lower3pt\hbox{$\scriptscriptstyle\frown$}} 
\over h} }}}_{{n_{\rm r}},m,i}})$, ${{{\bf{\mathord{\buildrel{\lower3pt\hbox{$\scriptscriptstyle\smile$}} 
\over H} }}}_{{n_{\rm r}},n,j}} = {{\bf{\Pi }}^{{{\cal V}_n}(j)}}{\mathop{\rm diag}\nolimits} ({{{\bf{\mathord{\buildrel{\lower3pt\hbox{$\scriptscriptstyle\smile$}} 
\over h} }}}_{{n_{\rm r}},n,j}})$, ${{{\bf{\mathord{\buildrel{\lower3pt\hbox{$\scriptscriptstyle\frown$}} 
\over h} }}}_{{n_{\rm r}},m,i}}$ and ${{{\bf{\mathord{\buildrel{\lower3pt\hbox{$\scriptscriptstyle\smile$}} 
\over h} }}}_{{n_{\rm r}},n,j}}$ are the vectors composed of diagonal elements of ${({\bf{\hat H}}_{{\rm{eff}}}^{({n_{\rm r}},i)})^{\rm H}}$ and ${\bf{\hat H}}_{{\rm{eff}}}^{({n_{\rm r}},j)}$ after cyclically shifting them according to their non-zero row index ${{{\cal U}_m}(i)}$ and non-zero column index ${{{\cal V}_n}(j)}$. It is worth noting that there are at most $KN$ non-zero elements in each row and column of ${\bf{\hat H}}_{{\rm{eff,d}}}^{({n_{\rm r}},{n_{\rm u}})}$, \emph{i.e.}, $|{{\cal U}_m}(i)| = |{{\cal V}_n}(j)| = KN$. By utilizing this sparse structure, the complexity of calculating ${\bf{\Xi }}$ can be greatly reduced.

The receiver undertakes an iterative process to improve channel estimation accuracy after obtaining the detected data ${\bf {\hat x}}_{\rm d}$. This iterative process involves channel estimation, interference cancellation, and data detection.  At the $k$-th iteration, the receiver removes the received data signal ${{\bf{y}}_{\rm{d}}}$ from the received signal to update the pilot signal ${\bf{\hat y}}_{\rm{p}}^{(k)} = {{\bf{y}}_{{\rm{DD}}}} - {{\bf{y}}_{\rm{d}}}$. The interference-canceled pilot signal is then used for the next channel estimation by replacing ${{\bf{y}}_{{\rm{DD}}}}$ in (13) with ${{{\bf{\hat y}}}_{\rm{p}}}$. After obtaining the updated channel ${\bf{\hat H}}_{{\rm{eff,d}}}^{(k)}$, the pilot signal is removed from the received signal to update data signal ${\bf{\hat y}}_{\rm{d}}^{(k)} = {{\bf{y}}_{{\rm{DD}}}} - {\bf{\hat H}}_{{\rm{eff,d}}}^{(k)}{{\bf{x}}_{\rm{p}}}$. The interference-canceled data signal ${\bf{\hat y}}_{\rm{d}}^{(k)}$ is used for the next data detection. 

\section{SBEE based DL Channel Prediction}
In this section, an SBEE channel predictor is proposed to address the problem of channel aging caused by fast time-varying channels, with the aid of estimated UL channel samples. As shown in Fig. 7, the channel samples of $N_{\rm t}$ UL OTFS frames are modeled by a set of Slepian sequences, and then SBEE predicts the channel of $N_{\rm f}$ DL OTFS frames by iteratively extrapolating the Slepian coefficients.
At the BS, the receiver outputs the smoothed UL channel ${{{\bf{\hat H}}}^{{\rm{UL}}}}$ which is then taken as the input of the SBEE predictor. ${{{\bf{\hat H}}}^{{\rm{UL}}}}$ can be fitted by $Q_{\rm{SP}}$ ($Q_{\rm{SP}} \ll MN$) Slepian sequences \cite{Qu2021} as
\begin{equation}
{{{\bf{\hat H}}}^{{\rm{UL}}}} = {{{\bf{\bar B}}}_{{\rm{SP}}}}{\bf{C}}_{{\rm{SP}}}^{{\rm{UL}}} + {\bf{W}}_{{\rm{SP}}}^{{\rm{UL}}},
\label{eq23}
\end{equation}
where ${{{\bf{\bar B}}}_{{\rm{SP}}}} = {{\bf{I}}_{{N_{\rm t}}}} \otimes {{\bf{B}}_{{\rm{SP}}}}$, ${\bf{C}}_{{\rm{SP}}}^{{\rm{UL}}}  \in {\mathbb{C}^{{N_{\rm t}}{Q_{{\rm{SP}}}} \times {N_{\rm r}}{N_{\rm u}}K}}$ is the Slepian coefficient matrix, ${\bf{W}}_{{\rm{SP}}}^{{\rm{UL}}}$ is the modeling error matrix, and ${{\bf{B}}_{{\rm{SP}}}}$ is the basis matrix of the Slepian sequences whose columns are the eigenvectors of matrix $\bm{\mathcal{D}}$
\begin{equation}
\bm{\mathcal{D}}{{\bf{b}}_{{\rm{SP}},q}} = {\chi_q} {{\bf{b}}_{{\rm{SP}},q}},
\label{eq24}
\end{equation}
where the entries of $\bm{\mathcal{D}}$ are
\begin{equation}
\bm{\mathcal{D}}(n,m) = \frac{{\sin (2\pi {f_{\max }}(n - m))}}{{\pi (n - m)}},
\label{eq25}
\end{equation}
 $n,m = 0, \cdots ,MN - 1$ and $q=0,\cdots,Q_{\rm{SP}}-1$. Then the UL Slepian coefficient matrix can be calculated as
\begin{equation}
{\bf{\hat C}}_{{\rm{SP}}}^{{\rm{UL}}} = {({\bf{\bar B}}_{{\rm{SP}}}^{\rm H}{{{\bf{\bar B}}}_{{\rm{SP}}}})^{ - 1}}{\bf{\bar B}}_{{\rm{SP}}}^{\rm H}{{{\bf{\hat H}}}^{{\rm{UL}}}}.
\label{eq26}
\end{equation}

\begin{figure}
\centerline{\includegraphics[width=0.5\textwidth]{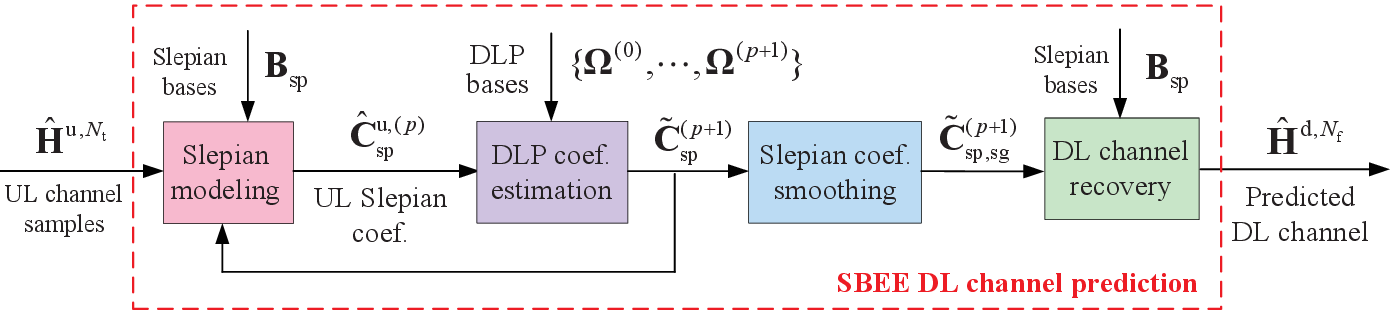}}
\caption{Block diagram of the proposed SBEE DL channel prediction scheme.}
\vspace{-0pt}
\label{fig7}
\end{figure}
\renewcommand{\algorithmicrequire}{\textbf{Input:}}
\renewcommand{\algorithmicensure}{\textbf{Output:}}
\begin{algorithm}[t]
\footnotesize
  \caption{\footnotesize {\bf {Algorithm 2~}} SBEE DL Channel Predictor for massive-MIMO-OTFS  } 
  \begin{algorithmic}[1]
    \Require 
      UL Slepian coefficient matrix  ${{{\bf{\bar C}}}_{{\rm{SP}}}}$, prediction length $N_{\rm f}$, prediction step size $\Delta $, DLP order $Q_{\rm{DLP}}$.
    \State Calculate the initial DLP matrix ${\bf{\Omega }}$ according to (26).
     \State Calculate the initial DLP coefficient ${{{\bf{\hat C}}}^{{\rm{DLP}},(0)}}$ according to (29).
       \While {$p < \frac{{{N_{\rm f}}}}{\Delta }$}
      \State $p=p+1$;
      \State Update the DLP matrix according to (30).
      \State Calculate the predicted Slepian coefficients according to (31).
     \State Update the Slepian coefficients according to (32).
     \State Smooth the Slepian coefficients by using VSG smoothing operation.
       \State Calculate the DLP coefficients according to (34).
    \EndWhile
    \Ensure Predicted $N_{\rm f}$-frame DL channel ${{{\bf{\hat H}}}^{{\rm{DL}},{N_{\rm f}}}}$ according to (35).   
  \end{algorithmic}
\end{algorithm}

The Slepian sequences are windowed (using rectangular windows) versions of infinite discrete prolate spheroidal sequences that are exactly band limited to the normalized Doppler frequency range $[ - {f_{\max }}/\Delta f,{f_{\max }}/\Delta f]$ \cite{Zemen2005}, where $\Delta f$ is the subcarrier spacing. As shown in \cite{Zemen2005}, the Slepian sequences outperform other commonly used BEMs in approximating a Jakes’ channel over a wide range of Doppler spreads for the same number of parameters. Hence, we adopt the Slepian sequences to model the estimated UL channels in this paper.
Our goal is to obtain the DL channels by predicting the DL Slepian coefficients based on the UL Slepian coefficient matrix ${\bf{\hat C}}_{{\rm{SP}}}^{{\rm{UL}}}$.
Note that the existing methods in \cite{Para2009,Para2015,Para2018,Wu2021,Kalman2000,Talaei2021,Prony2020,
MatrixPencil2022,Peng2019,Dong2019,Jiang2022,Chu2022,Mattu2022} directly predict DL channel based on the $MN{N_{\rm r}}{N_{\rm u}}{N_{\rm t}}K$ CIRs in ${{{\bf{\hat H}}}^{{\rm{UL}}}}$, which needs high-dimensional matrix operations or time-consuming offline training. 
The proposed SBEE predictor predicts the DL channel by using only $Q_{\rm SP}{N_{\rm r}}{N_{\rm u}}{N_{\rm t}}K$ Slepian coefficients, which reduces the complexity..

In this paper, the discrete Legendre polynomials (DLP) is adopted to fit Slepian coefficients due to its excellent fitting performance and good numerical stability \cite{Senol2021}. The Slepian coefficients can be accurately extrapolated by dynamically fitting the Slepian coefficients with DLP and solving a small number of DLP coefficients, which reduces the complexity of predicting the Slepian coefficients. 
The $(q+1)$-order ($q\ge1$) DLP is defined as 
\begin{equation}
{\varphi _{q + 1}}[t_n] = \frac{{(2q + 1)t}}{{q + 1}}{\varphi _q}[t_n] - \frac{q}{{q + 1}}{\varphi _{q - 1}}[t_n].
\label{eq27}
\end{equation}

Its first two terms are ${\varphi  _0}[t_n] = 1$ and ${\varphi  _1}[t_n] = t_n$ with ${t_n} = \frac{{2(n - 1)}}{{{N_{\rm t}} - 1}} - 1$, $n =1,\cdots,N_{\rm t}$.

Define ${{{\bf{\bar c}}}_{{\rm{SP}},{n_{\rm t}}}} = {[{[{\bf{\hat C}}_{{\rm{SP}}}^{{\rm{UL}}}]_{({n_{\rm t}} - 1){Q_{{\rm{SP}}}} + 1}}, \cdots ,{[{\bf{\hat C}}_{{\rm{SP}}}^{{\rm{UL}}}]_{{n_{\rm t}}{Q_{{\rm{SP}}}}}}]^{\rm H}}$, and let ${{{\bf{\bar C}}}_{{\rm{SP}}}} = {[{{{\bf{\bar c}}}_{{\rm{SP,1}}}}, \cdots ,{{{\bf{\bar c}}}_{{\rm{SP,}}{N_{\rm t}}}}]^{\rm H}} \in {\mathbb{C}^{{N_{\rm t}} \times {N_{\rm r}}{N_{\rm u}}K{Q_{{\rm{SP}}}}}}$, the Slepian coefficient ${{\bar C}_{{\rm{SP}}}}({n_{\rm t}},k)$ can be fitted by
\begin{equation}
{{\bar C}_{{\rm{SP}}}}({n_{\rm t}},k) = \sum\nolimits_{q = 0}^{{Q_{{\rm{SP}}}} - 1} {{{c_{k,q}^{{\rm{DLP}}}}}{\varphi _q}[t_{n_{\rm t}}]}  + {W_{\rm DLP}}[{n_{\rm t}},k],
\label{eq28}
\end{equation}
where ${{\bar C}_{{\rm{SP}}}}({n_{\rm t}},k)$ ($1\le {n_{\rm t}}\le {N_{\rm t}}, 1\le k\le {N_{\rm r}}{N_{\rm u}}K{Q_{\rm SP}}$) is the ${n_{\rm t}}$-th row and the $k$-th column of ${{{\bf{\bar C}}}_{{\rm{SP}}}}$, ${c_{k,q}^{{\rm{DLP}}}}$ is the DLP coefficients and ${W_{\rm DLP} }[{n_{\rm t}},k]$ is the DLP fitting error.
The vectorized version of (28) can be written as
\begin{equation}
{{{\bf{\bar C}}}_{{\rm{SP}}}} = {\bf{\Omega }}{{\bf{C}}^{{\rm{DLP}}}} + {\bf{{W} _{\rm DLP}}},
\label{eq29}
\end{equation}
where ${\bf{\Omega }} = [{{\bm{\varphi }}_0}, \cdots ,{{\bm{\varphi }}_{{Q_{\rm DLP}} - 1}}]\in {{\mathbb{C}}^{{{N}_{\text{t}}}\times {{Q}_{\text{DLP}}}}}$ is the $Q_{\rm{DLP}}$-order DLP basis matrix with ${{\bm{\varphi }}_q} = {[{\varphi _q}[t_1], \cdots ,{\varphi _q}[t_{N_{\rm t}}]]^{\rm H}}$, ${{\bf{C}}^{{\rm{DLP}}}} = [{\bf{c}}_1^{{\rm{DLP}}}, \cdots ,{\bf{c}}_{{N_{\rm r}}{N_{\rm u}}K{Q_{{\rm{SP}}}}}^{{\rm{DLP}}}]\in {{\mathbb{C}}^{{{Q}_{\text{DLP}}}\times {{N}_{\text{r}}}{{N}_{\text{u}}}K{{Q}_{\text{SP}}}}}$ with ${\bf{c}}_k^{{\rm{DLP}}} = {[c_{k,0}^{{\rm{DLP}}}, \cdots ,c_{k,{Q_{{\rm{DLP}}-1}}}^{{\rm{DLP}}}]^{\rm H}}\in {{\mathbb{C}}^{{{Q}_{\text{DLP}}}}}$ is the DLP coefficient matrix and ${{\bf W}_{\rm DLP}}$ is the modeling error matrix introduced by DLP.

In the ($0$)-th iteration, the initial DLP fitting coefficient can be calculated as
\begin{equation}
{{{\bf{\hat C}}}^{{{\rm{DLP}},(0)}}}= {({{\bf{\Omega }}^{\rm H}}{\bf{\Omega }})^{ - 1}}{{\bf{\Omega }}^{\rm H}}{{{\bf{\bar C}}}_{{\rm{SP}}}}.
\label{eq30}
\end{equation}

To exploit the global temporal correlation of the $N_{\rm t}$ UL channel samples and the predicted channels for high accuracy long-term prediction,  the DLP is updated in each iteration according to the previous prediction. In the ($p+1$)-th ($p\ge 0$) iteration, a new DLP matrix is constructed as
\begin{equation}
{{\bf{\Omega }}^{(p + 1)}} = [ {\bm {\varphi} _0^{(p + 1)}, \cdots ,\bm {\varphi }_{{Q_{{\rm{DLP}}}} - 1}^{(p + 1)}} ].
\label{eq31}
\end{equation}
where $\bm{\varphi }_q^{(p + 1)} = {[{\varphi _q}[{t_1}], \cdots ,{\varphi _q}[{t_{{N_{\rm t}} + p\Delta }}]]^{\rm T}} \in {\mathbb{R}^{{N_{\rm t}} + p\Delta }}$, ${t_n} = \frac{{2(n - 1)}}{{{N_{\rm t}} + p\Delta  - 1}} - 1$, and $\Delta$ ($\Delta \ge 1$) is the step size of prediction. The DL Slepian coefficient prediction model is derived as 
\begin{equation}
\hat {\bar C}_{{\rm{SP}}}^{(p+1)}({N_{\rm t}} + p\Delta ,k) = \sum\nolimits_{q = 0}^{{Q_{{\rm{DLP}}}} - 1} {\hat c_{k,q}^{{\rm{DLP,}}(p)}{\varphi _q^{(p+1)}}[t_{{N_{\rm t}} + p\Delta} ]}.
\label{eq32}
\end{equation}

The Slepian coefficients are then updated by adding the predicted Slepian coefficients in the ($p+1$)-th iteration to the previous Slepian coefficients
\begin{equation}
{\bf{\tilde C}}_{{\rm{SP}}}^{(p + 1)} = {[{({\bf{\bar C}}_{{\rm{SP}}}^{})^{\rm H}}, \cdots ,{({[{\bf{\hat {\bar C}}}_{{\rm{SP}}}^{(p + 1)}]_{{N_{\rm t}+1}:{N_{\rm t}} + p\Delta }})^{\rm H}}]^{\rm H}}.
\label{eq33}
\end{equation}

To improve the prediction accuracy of DL Slepian coefficients, the Slepian coefficient ${\bf{\tilde C}}_{{\rm{SP}}}^{(p + 1)}$ is smoothed by the SG smoothing technique introduced in Section III-D.
\begin{equation}
{\bf{\tilde C}}_{{\rm{SP,SG}}}^{(p + 1)} = {\mathop{\rm SG}\nolimits} ({\bf{\tilde C}}_{{\rm{SP}}}^{(p + 1)})
\label{eq34}
\end{equation}
where ${\mathop{\rm SG}\nolimits} ( \cdot )$ denotes the SG smoothing operation.

Then the DLP coefficient matrix in the ($p+1$)-th iteration ${{{\bf{\hat C}}}^{{\rm{DLP}},(p + 1)}}$ can be updated by
\begin{equation}
{{{\bf{\hat C}}}^{{\rm{DLP}},(p + 1)}} = {({{\bf{\Omega }}^{(p + 1)}})^\dag }{\bf{\tilde C}}_{{\rm{SP}},{\rm{SG}}}^{(p + 1)}.
\label{eq35}
\end{equation}

By iteratively performing (31)-(35), we can obtain the predicted $N_{\rm f}$-frame ($N_{\rm f}=p\Delta$) DL Slepian coefficient matrix. Finally, the predicted $N_{\rm f}$-frame DL channels can be recovered by the predicted Slepian coefficient as
\begin{equation}
{{{\bf{\hat H}}}^{{\rm{DL}},{N_{\rm f}}}} = ({{\bf{I}}_{{N_{\rm f}}}} \otimes {{\bf{B}}_{{\rm{SP}}}}){\bf{\mathord{\buildrel{\lower3pt\hbox{$\scriptscriptstyle\frown$}} 
\over C} }}_{{\rm{SP}}}^{(p + 1)},
\label{eq36}
\end{equation}
where ${\bf{\mathord{\buildrel{\lower3pt\hbox{$\scriptscriptstyle\frown$}} 
\over C} }}_{{\rm{SP}}}^{(p + 1)} = {[{({\bf{\mathord{\buildrel{\lower3pt\hbox{$\scriptscriptstyle\frown$}} 
\over c} }}_{{\rm{SP,1}}}^{(p + 1)})^{\rm H}}, \cdots ,{({\bf{\mathord{\buildrel{\lower3pt\hbox{$\scriptscriptstyle\frown$}} 
\over c} }}_{{\rm{SP,}}{N_{\rm f}}}^{(p + 1)})^{\rm H}}]^{\rm H}}$ and ${\bf{\mathord{\buildrel{\lower3pt\hbox{$\scriptscriptstyle\frown$}} 
\over c} }}_{{\rm{SP,}}{n_{\rm f}}}^{(p + 1)} = {({[{\bf{ {\tilde C}}}_{{\rm{SP,SG}}}^{(p + 1)}]_{{N_{\rm t}} + {n_{\rm f}},1:{Q_{{\rm{SP}}}}}})^{\rm H}}, \cdots ,$\\${({[{\bf{ {\tilde C}}}_{{\rm{SP,SG}}}^{(p + 1)}]_{{N_{\rm t}} + {n_{\rm f}},({N_{\rm r}}{N_{\rm u}}K - 1){Q_{{\rm{SP}}}} + 1:({N_{\rm r}}{N_{\rm u}}K){Q_{{\rm{SP}}}}}})^{\rm H}}$.

The pseudo-code of the SBEE channel prediction algorithm is summarized in \textbf{Algorithm 2}.

\vspace{-0pt}
\section{Performance and Complexity Analysis}
In this section, the asymptotic error analysis for UL channel modeling and DL channel prediction are provided. Furthermore, the complexity of UL channel estimation and DL channel prediction is analyzed.

\vspace{-0pt}
\subsection{Asymptotic Performance Analysis}
The asymptotic error of CE-BEM based UL channel modeling is given in \textbf{Lemma 1}.

\theorembodyfont{\upshape}
\newtheorem*{proof}{\indent \it{Proof: }}
\newtheorem{lemma}{Lemma}
\begin{lemma}
For a given value $N$ and user velocity $v$, the asymptotic modeling error of the UL channel yields
\begin{equation}
\mathop {\lim }\limits_{M \to \infty } {{\mathop{\rm MSE}\nolimits} _{{\rm{mod}}}} = \mathop {\lim }\limits_{M \to \infty } \frac{{||{{\bf{H}}^{{\rm{UL}},{N_{\rm t}}}} - {{{\bf{\hat H}}}^{{\rm{UL}},{N_{\rm t}}}}||_{\rm F}^2}}{{||{{\bf{H}}^{{\rm{UL}},{N_{\rm t}}}}||_{\rm F}^2}} = 0
\label{eq37}
\end{equation}
when $Q_{\rm s}=N_{\rm r}$, where ${{{\bf{H}}^{{\rm{UL}},{N_{\rm t}}}}}$ and ${{{{\bf{\hat H}}}^{{\rm{UL}},{N_{\rm t}}}}}$ are the true and BEM modeled $N_{\rm t}$-frame UL channel matrices, respectively.
\end{lemma}
\begin{proof}
{The proof can be found in Appendix B.$\hfill\blacksquare$}
\end{proof}

Based on \textbf{Lemma 1}, the asymptotic error of the SBEE channel predictor is shown in \textbf{Theorem 1}.
\newtheorem{theorem}{Theorem}
\begin{theorem}
For a given value $N$ and user velocity $v$, the asymptotic performance of the SBEE predictor yields
\begin{equation}
\mathop {\lim }\limits_{{\sigma ^2} \to 0,M \to \infty } \frac{{||{{\bf{H}}^{{\rm{DL}},{N_{\rm f}}}} - {{{\bf{\hat H}}}^{{\rm{DL}},{N_{\rm f}}}}||_{\rm F}^2}}{{||{{\bf{H}}^{{\rm{DL}},{N_{\rm f}}}}||_{\rm F}^2}} = 0,
\label{eq38}
\end{equation}
under the condition of perfect UL channel estimation, where ${{{\bf{H}}^{{\rm{DL}},{N_{\rm f}}}}}$ and ${{{{\bf{\hat H}}}^{{\rm{DL}},{N_{\rm f}}}}}$ are the true and predicted $N_{\rm f}$-frame DL channel matrices, respectively.
\end{theorem}
\begin{proof}
{The proof can be found in Appendix C.$\hfill\blacksquare$}
\end{proof}
\emph{Remarks:} 
According to CS theory \cite{Tropp2007}, $K$-sparse signals can be accurately recovered from $G > K{\rm{log}}(L{\rm{/}}K)$ linear measurements by using greedy algorithm. This reveals that the supports of non-zero UL channels can be accurately recovered by using the proposed VBL-SOMP algorithm when the number of pilots is large enough. Based on \textbf{Lemma 1}, as noise power ${\sigma ^2} \to 0$ and $M \to \infty $, the condition for perfect UL channel estimation can be achieved. 

\vspace{-0pt}
\subsection{Complexity Analysis}
The symbolic and numerical complexity of UL channel estimation and DL channel prediction are analyzed.

\emph{1) Complexity of BEM UL channel estimation.}
The complexity of BEM UL channel estimation mainly lies in VBL-SOMP algorithm, SG channel smoothing and data detection. In \textbf{Algorithm 1}, the complexity of computing block correlation in step 3 is $O(G{N_{\rm r}}{N_{\rm u}}L{Q_{\rm{s}}}Q({N_{\rm u}}(K - {K_{\rm C}}) + {K_{\rm C}}))$. The complexity of recovering the common blocks in step 7 is $O({({N_{\rm u}}{Q_{\rm s}}{K_{\rm C}})^3} + 2G{N_{\rm r}}{({N_{\rm u}}{K_{\rm C}}{Q_{\rm s}})^2} + G{N_{\rm r}}{N_{\rm u}}L{Q_{\rm s}}Q)$. The complexity of recovering individual blocks in step 13 is $O({({N_{\rm u}}{Q_{\rm{s}}}(K - {K_{\rm C}}))^3} + 2G{N_{\rm r}}N_{\rm u}^4Q_{{\rm{s}}}^2{(K - {K_{\rm C}})^2} + G{N_{\rm r}}{N_{\rm u}}L{Q_{\rm{s}}}Q)$. In the SG channel smoothing process, the complexity of calculating the polynomial coefficients in (19) is $O(MN({({Q_{{\rm sg}}} + 1)^3} + 2{({Q_{{\rm sg}}} + 1)^2}(2{N_{{\rm{sg}}}} + 1)))$. The complexity of calculating (20) is $O(MN({Q_{{\rm{sg}}}} + 1)(2{N_{{\rm{sg}}}} + 1){N_{\rm u}}{Q_{\rm{s}}}K)$. For data detection, the complexity of calculating ${\bf{\Xi }}$ is $O({N_{\rm d}}{N_{\rm r}}{N^2}N_{\rm u}^2{K^2})$, and the complexity of CG algorithm is $O({I_{\rm CG}}N_{\rm u}^2N_{\rm d}^2)$, where ${I_{\rm CG}}$ is the number of iterations of CG algorithm. Define the maximum number of iterations for BEM UL channel estimation as $I_{\rm max}$. Note that ${K_{\rm C}^3},K^3,{Q_{{\rm{s}}}^3},Q,{N_{\rm u}^3},{Q_{{{\rm{sg}}}}^3}$ and ${N_{\rm sg}}$ is much smaller than $MN$, $N_{\rm r}$ and $N_{\rm d}$. The total complexity of UL channel estimation can be summarized as $O({I_{\max }}(2G{N_{\rm r}}N_{\rm u}^4Q_{\rm s}^2{(K - {K_{\rm C}})^2} + G{N_{\rm r}}{N_{\rm u}^2}L{Q_{\rm s}}QK + MNQ_{\rm sg}^3 + {N_{\rm d}}{N_{\rm r}}{N^2}N_{\rm u}^2{K^2} + {I_{\rm {CG}}}N_{\rm u}^2N_{\rm d}^2))$.

\emph{ 2) Complexity of DL channel prediction.} 
The complexity of SBEE channel prediction mainly lies in calculations of (26), (34), (35) and (36). The complexity of calculating UL Slepian coefficients in (26) is $O( MN{N_{\rm r}}{N_{\rm u}}{N_{\rm t}}K)$. The complexity of Slepian coefficient smoothing in step (34) is $O(({N_{\rm t}} + {N_{\rm f}})({({Q_{{\rm sg}}} + 1)^3} + 2{({Q_{{\rm sg}}} + 1)^2}(2{N_{{\rm sg}}} + 1) + ({Q_{{\rm sg}}} + 1)(2{N_{{\rm sg}}} + 1){N_{\rm r}}{N_{\rm u}}K{Q_{{\rm{SP}}}}))$. The complexity of updating DL Slepian coefficients in step (35) is $O(Q_{{\rm{DLP}}}^3 + 2{N_{\rm t}}Q_{{\rm{DLP}}}^2 + {N_{\rm t}}{N_{\rm r}}{N_{\rm u}}K{Q_{{\rm{SP}}}}{Q_{{\rm{DLP}}}})$. The complexity of updating DL channel in (36) is $O(MN{N_{\rm r}}{N_{\rm u}}K{N_{\rm f}}{Q_{{\rm{SP}}}})$. Note that $Q_{{\rm{SP}}}^3$ and $Q_{{\rm{DLP}}}^3$ is much smaller than $MN$. Thus, the DL channel updation dominates the complexity of SBEE channel predictor. The total complexity of DL channel prediction can be summarized as $O(MN{N_{\rm r}}{N_{\rm u}}K{N_{\rm f}}{Q_{{\rm{SP}}}})$.

The symbolic and numerical complexity of the proposed BEM UL channel estimator and SBEE DL channel predictor are shown in Table I. The existing SOMP based channel estimator \cite{Shen2019} and vectorized Prony channel predictor \cite{Prony2020} are selected for complexity comparison. Numerical complexity analysis shows that the proposed UL channel estimator achieves 20\% complexity reduction compared to existing SOMP based channel estimator \cite{Shen2019}, and the proposed SBEE channel predictor achieves up to 70\% complexity reduction compared to existing vector Prony channel predictor \cite{Prony2020}.

\begin{table*}
\centering
\renewcommand{\arraystretch}{1.2}
\scriptsize
\caption{Symbolic and Numerical Complexity of UL Channel Estimation and DL Channel Prediction. $G=32$, $M=128$, $N=8$, $N_{\rm r}=64$, $N_{\rm u}=2$, $Q_{\rm s}=32$, $L=64$, $K=4$, $Q_{\rm {SP}}=5$, $N_{\rm{sg}}=5$, $Q=3$, $I_{\rm {CG}}=20$, $I_{\rm{max}}=4$, $N_{\rm{vp}}=5$ is the order of vector Prony predictor \cite{Prony2020}.}
\begin{tabular}{|m{1.8cm}<{\centering}|m{3.8cm}<{\centering}|m{1.3cm}<{\centering}|m{3.8cm}<{\centering}|m{3.6cm}<{\centering}|} 
\hline
\multirow{2}{*}{Schemes } & \multicolumn{3}{l|}{\qquad\qquad\qquad\qquad Symbolic complexity of UL channel estimation (per iteration)}                                                              &\multirow{2}{*}{Normalized numerical complexity } \\ \cline{2-4} 
                 & \multicolumn{1}{l|}{\quad\quad BEM coefficient estimation  } & \multicolumn{1}{l|}{VSG smoothing}                  &Data detection                   &  \\ \hline
                  SOMP \cite{Shen2019}& \multicolumn{1}{l|}{$O(2G{N_{\rm r}}{N_{\rm u}}{Q_{\rm{s}}}({K^2}{N_{\rm u}}{Q_{\rm{s}}} + LQ))$} & \multicolumn{1}{l|}{\multirow{2}{*}{$O(MNQ_{{\rm{sg}}}^3)$}} & \multirow{2}{*}{$O({N_{\rm d}}{N_{\rm r}}{N^2}N_{\rm u}^2{K^2}+ {I_{\rm {CG}}}N_{\rm u}^2N_{\rm d}^2)$} & 1.2 \\ \cline{1-2} \cline{5-5} 
                  VBL-SOMP& \multicolumn{1}{l|}{$O(2G{N_{\rm r}}N_{\rm u}^4{Q_{\rm s}^2}(K-K_{\rm C})^2)$} & \multicolumn{1}{l|}{}                 &                   &1  \\ \hline \hline
                  Schemes& \multicolumn{3}{l|}{\qquad\qquad\qquad \qquad \qquad  Symbolic complexity of DL channel prediction}                            &Normalized numerical complexity  \\ \hline
 Vector Prony \cite{Prony2020}& \multicolumn{3}{l|}{\qquad\qquad\qquad \qquad \qquad$O(N_{{\rm{vp}}}^2MN{N_{\rm r}}{N_{\rm u}}K + {N_{{\rm{vp}}}}MN{N_{\rm r}}{N_{\rm u}}{N_{\rm f}}K)$}                            &3.4  \\ \hline
 SBEE & \multicolumn{3}{l|}{\qquad\qquad\qquad \qquad \qquad\qquad\qquad\qquad$O(MN{N_{\rm r}}{N_{\rm u}}K{N_{\rm f}}{Q_{{\rm{SP}}}} )$}                            &1  \\ \hline
\end{tabular}
\end{table*}
\vspace{-0pt}

\begin{table}[htbp]
\renewcommand{\arraystretch}{1.2}
\centering
\caption{Simulation Parameters}
\begin{tabular}{|c|c|}
\hline
\textbf{System parameter} & \textbf{Value} \\
\hline
Carrier frequency ($f_{\rm c}$) & 3 GHz \\
\hline
Subcarrier spacing ($\Delta f$) & 30 KHz \\
\hline
Bandwidth ($B_{\rm w}$) & 3.84 MHz \\
\hline
Number of delay bins ($M$) & 128 \\
\hline
Number of Doppler bins ($N$) & 1, 4, 8, 16 \\
\hline
Number of BS antennas ($N_{\rm r}$)& 64 \\
\hline
Number of users ($N_{\rm u}$)& 2 \\
\hline
The total number of channel paths ($L$)& 64 \\
\hline
Number of non-zero paths ($K$)& 4 \\
\hline
Number of commom paths ($K_{\rm C}$)& 1 \\
\hline
Order of CE-BEM ($Q$)& 3 \\
\hline
Number of Slepian sequences ($Q_{\rm{SP}}$)& 5 \\
\hline
Order of DLP ($Q_{\rm{DLP}}$)& 5 \\
\hline
Order of SG smoother ($Q_{\rm{sg}}$)& 5 \\
\hline
Number of samples used in SG smoother ($N_{\rm{sg}}$)& 5 \\
\hline
Number of non-zero pilots per frame ($G$)& 32 \\
\hline
Modulation scheme & QPSK \\
\hline
Channel model & 5G TDL-B \cite{5g-channel} \\
\hline
\end{tabular}
\label{tab3}
\end{table}

\vspace{-5pt}
\section{Simulation Results}

In this section, the performance of the proposed BEM UL channel estimator and SBEE DL channel predictor are verified by extensive Monte Carlo simulation results. 

\vspace{-5pt}
\subsection{Simulation Setup}
The simulation parameters are set as follows. The carrier frequency is set to 3 GHz with subcarrier spacing of 30 KHz. Each OTFS frame includes 128 delay bins and 8 Doppler bins. The number of BS antennas and users are $N_{\rm r}=64$ and $N_{\rm u}=2$, respectively. The 5G TDL-B channel model \cite{5g-channel} with channel length $L=64$ and Jakes' Doppler spectrum are used to model doubly selective channels. The number of nonzero paths is $K=4$.
The orders of CE-BEM, Slepian sequences, DLP and SG channel smoother are $Q=3$, $Q_{\rm{SP}}=5$, $Q_{\rm{DLP}}=5$, and $Q_{\rm{sg}}=5$, respectively. In all figures, except for Fig. 11, the number of non-zero pilots in each OTFS frame is fixed at $G=32$, resulting in a pilot overhead of 16\%. quadrature phase shift keying (QPSK) modulation technique is adopted. The values of the key system parameters are summarized in Table III.

\begin{figure}[htbp]
\centering	
\begin{minipage}[t]{0.46\textwidth}
	\centering
	\includegraphics[width=1\textwidth]{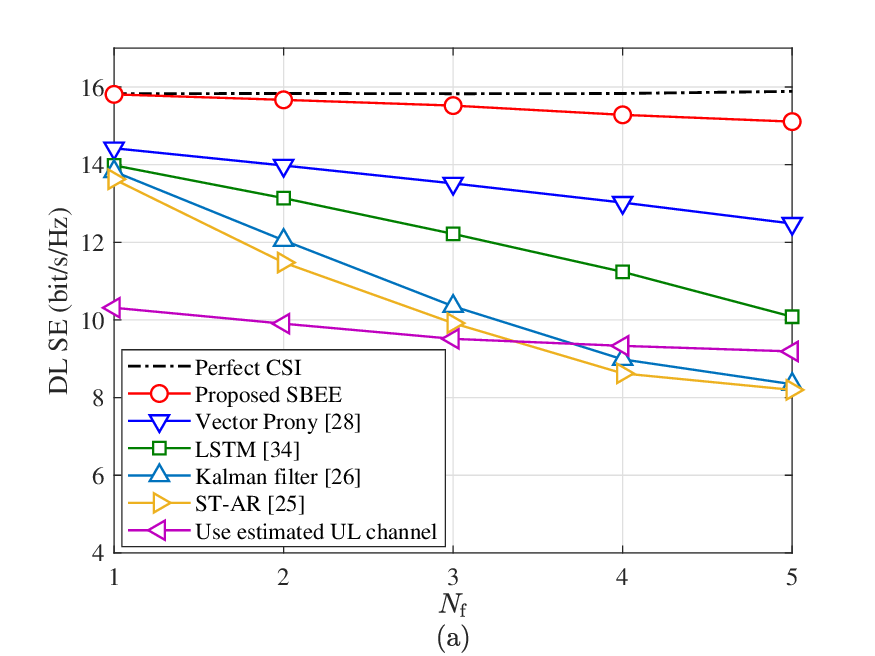}
		\vspace{-10pt} 
\end{minipage}
\begin{minipage}[t]{0.46\textwidth}
	\centering
	\includegraphics[width=1\textwidth]{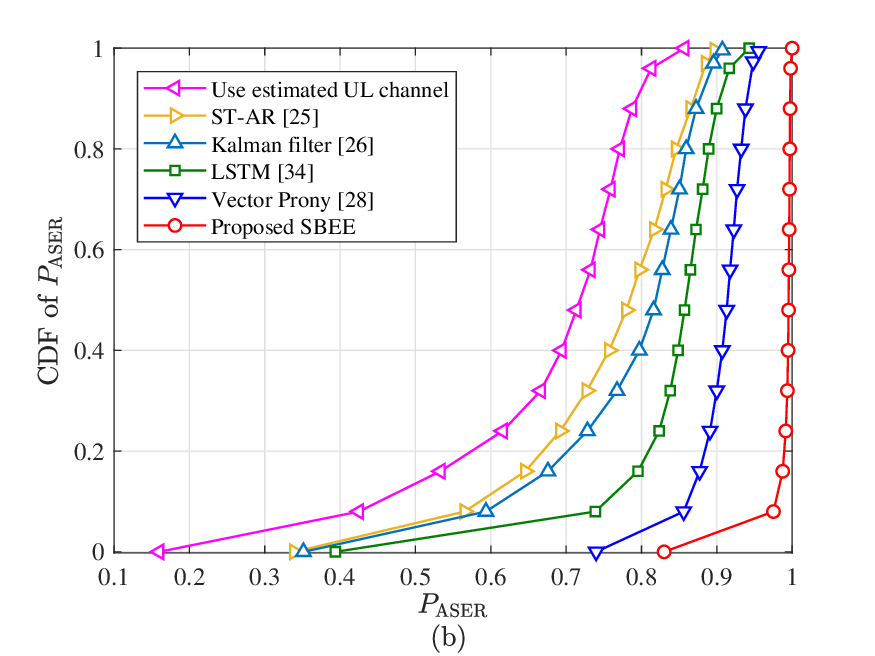}
	\vspace{-0pt} 
\end{minipage}
\vspace{-5pt}
\caption{(a) DL SE vs. the number of predicted OTFS frames $N_{\rm f}$, and (b) CDF of DL ASER $P_{\rm{ASER}}$ with $N_{\rm f}=2$, with $\rm{SNR} = 15$ dB, $N_{\rm r}=64$, $Q_{\rm s}=N_{\rm r}$, $N_{\rm u}=2$ and $v= 120$ km/h.\vspace{-5pt}}	
\label{fig8}
\end{figure}

The SOMP channel estimator \cite{Shen2019}, BSOMP channel estimator \cite{Gong2017}, spatio-temporal autoregression \mbox{(ST-AR)} predictor \cite{Wu2021}, Kalman filter based predictor \cite{Kalman2000}, vector Prony predictor \cite{Prony2020} 
and LSTM predictor \cite{Mattu2022} are selected as competitors. The genie-aided LS algorithm is also provided as a baseline scheme that uses LS algorithm to estimate  channel with the perfect knowledge of non-zero channel path indices and common path indices. For all channel predictors, the number of UL channel frames use for training is set to $N_{\rm t}=5$. 
The order of ST-AR and vector Prony predictors are also set as $Q_{\rm{DLP}}$ for a fair comparison. The LSTM predictor consists of 1 LSTM layer with 100 hidden units and one fully connected layer \cite{Mattu2022}.  
The logarithmic NMSE of channel estimation and channel prediction are defined as ${\rm{NMS}}{{\rm{E}}_{{\rm{CE}}}} = 10{{\log}_{10}}\mathbb{E} \left\{ {\frac{{\left\| {{{\bf{H}}^{{\rm{UL}}}} - {{{\bf{\hat H}}}^{{\rm{UL}}}}} \right\|_{\rm F}^2}}{{\left\| {{{\bf{H}}^{{\rm{UL}}}}} \right\|_{\rm F}^2}}} \right\}$ and ${\rm{NMS}}{{\rm{E}}_{{\rm{CP}}}} = 10{{\log}_{10}} \mathbb{E}\left\{ {\frac{{\left\| {{{\bf{H}}^{{\rm{DL}}}} - {{{\bf{\hat H}}}^{{{{\rm{DL}}}}}}} \right\|_{\rm F}^2}}{{\left\| {{{\bf{H}}^{{\rm{DL}}}}} \right\|_{\rm F}^2}}} \right\}$, respectively, where ${{{\bf{H}}^{{\rm{UL}}}}}$ and ${{{\bf{H}}^{{\rm{DL}}}}}$ are the true UL and DL channel matrices, ${{{{\bf{\hat H}}}^{{\rm{UL}}}}}$ and ${{{{\bf{\hat H}}}^{{\rm{DL}}}}}$ are the estimated UL and DL channel matrices.
The DL SE is calculated as \cite{Prony2020}
\begin{equation}
{R_{\rm SE}} = \sum\limits_{{n_{\rm u}} = 1}^{{N_{\rm u}}} {\mathbb{E}\left\{ {{{\log}_2} \left( {1 + \frac{{||{\bf{\bar h}}_{{n_{\rm u}}}^{\rm T}{{\bf{w}}_{{n_{\rm u}}}}||_2^2}}{{\sum\nolimits_{n \ne {n_{\rm u}}}^{{N_{\rm u}}} {||{\bf{\bar h}}_n^{\rm T}{{\bf{w}}_n}||_2^2}  + \sigma _{{n}}^2}}} \right)} \right\}},
\label{eq39}
\end{equation}
where ${{{\bf{\bar h}}}_{{n_{\rm u}}}}$ is the predicted DL channel, ${{{\bf{w}}_{{n_{\rm u}}}}}$ is the zero-forcing precoding vector obtained using the predicted DL channel \cite{Jiang2022}. 
The average SE ratio (ASER) is used as the performance evaluation criterion of DL channel predictors, which is defined as
\begin{equation}
{P_{\rm{ASER}}} = \mathbb{E}{{\{ {{\hat R}_{{\rm{SE}}}}\} } \mathord{\left/
 {\vphantom {{\{ {{\hat R}_{{\rm{SE}}}}\} } {E\{ {R_{{\rm{SE}}}}\} }}} \right.
 \kern-\nulldelimiterspace} {\mathbb{E}\{ {R_{{\rm{SE}}}}\} }},
\label{eq40}
\end{equation}
where ${\mathbb{E}\{ {{\hat R}_{{\rm{SE}}}}\} }$ and ${\mathbb{E}\{ {R_{{\rm{SE}}}}\} }$ denote the average SE obtained by using the predicted DL channels and the perfect DL channels, respectively.

\vspace{-5pt}
\subsection{DL Channel Prediction Performance}
In this subsection, we investigate the NMSE performance and DL SE of the proposed SBEE channel predictor at different SNRs and velocities.

\begin{figure}[htbp]
\centering	
\begin{minipage}[t]{0.46\textwidth}
	\centering
	\includegraphics[width=1\textwidth]{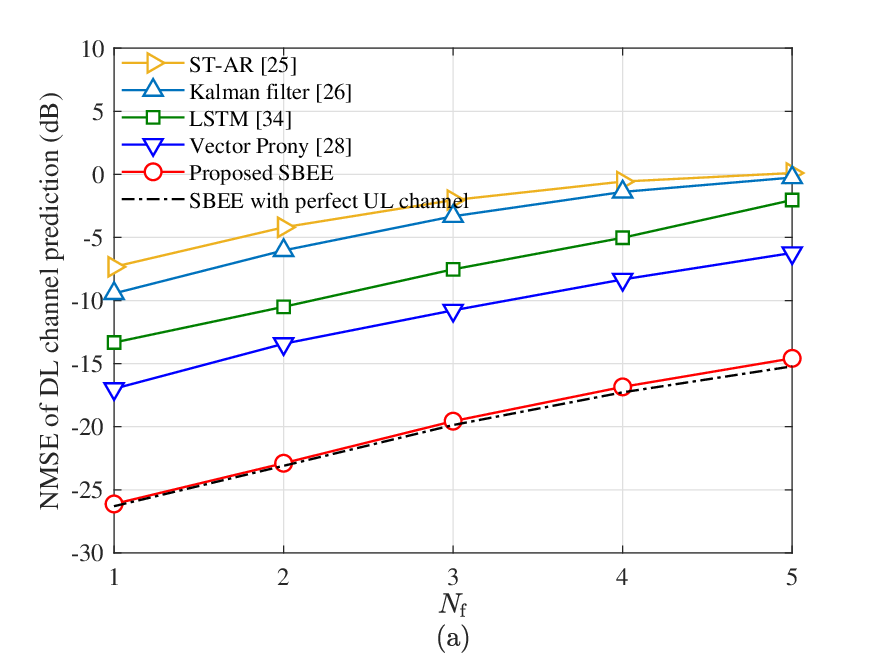}
		\vspace{-10pt} 
\end{minipage}
\begin{minipage}[t]{0.46\textwidth}
	\centering
	\includegraphics[width=1\textwidth]{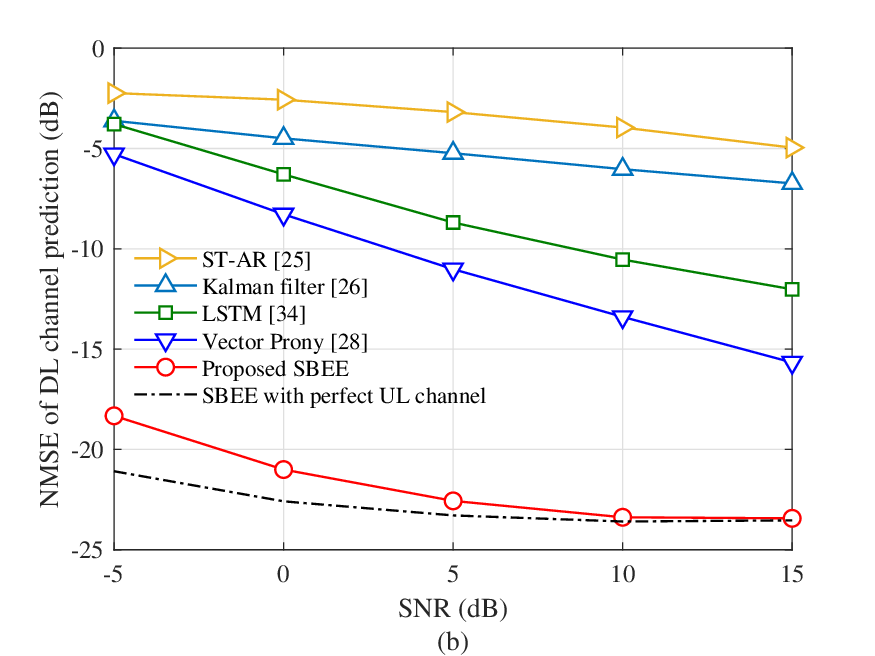}
	\vspace{-0pt} 
\end{minipage}
\vspace{-5pt}
\caption{NMSE of DL channel prediction as a function of (a) $N_{\rm f}$ with $\rm{SNR}=10$ dB, and (b) SNR with $N_{\rm f}=2$. \vspace{-10pt}}
\label{fig9}
\end{figure}

Fig. 8 a) shows the DL SE achieved by the proposed SBEE channel predictor and existing ST-AR predictor \cite{Wu2021}, Kalman filter based predictor \cite{Kalman2000}, vector Prony predictor \cite{Prony2020} 
and LSTM predictor \cite{Mattu2022}. The upper bound of DL SE is achieved by the scheme with perfect CSI available. We can observe that the proposed SBEE predictor could approach the upper bound of DL SE at different prediction frame length. Even at $N_{\rm f} = 5$, the SBEE predictor can achieve nearly 94\% DL SE of the upper bound of SE. Existing channel prediction schemes in \cite{Wu2021}, \cite{Kalman2000}, \cite{Prony2020} and \cite{Mattu2022} lead to large SE loss, especially when the number of predicted OTFS frames is large. Fig. 8 b) presents the cumulative distribution function (CDF) of the DL SER corresponding to different channel predictors. It can be seen that the DL SER achieved by the proposed SBEE predictor is very close to 1 with high probability. This means that the SBEE predictor has high prediction accuracy. The superior performance of proposed SBEE channel predictor lies in its utilization of global temporal correlation between the UL channel estimates and the subsequent DL channel predicts in each iteration. In contrast, the existing methods in \cite{Wu2021}, \cite{Kalman2000}, \cite{Prony2020}, and \cite{Mattu2022} could not make use of the already obtained DL channel predicts to enhance the subsequent channel prediction.

\begin{figure}[htbp]
\centerline{\includegraphics[width=0.46\textwidth]{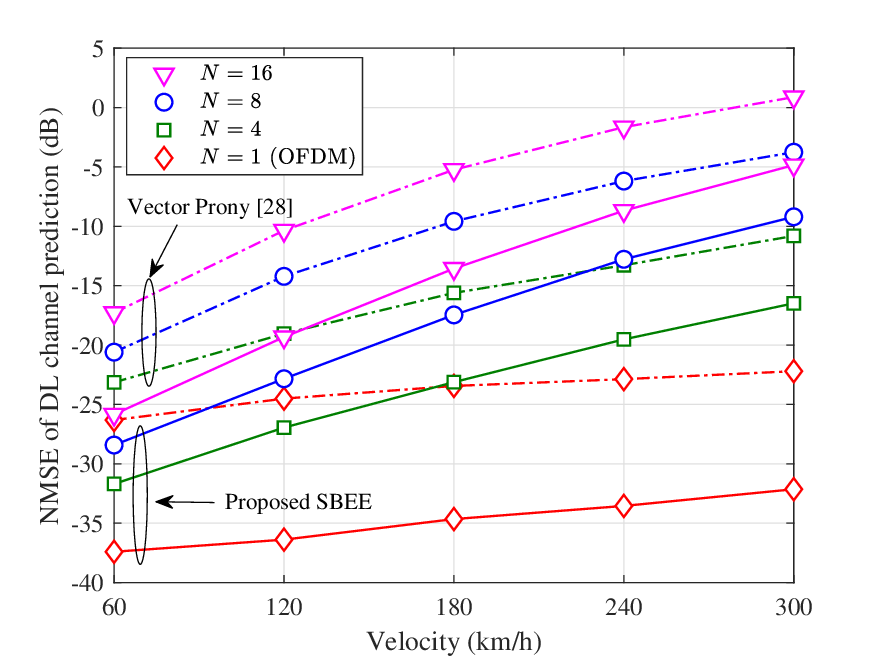}}
\caption{Impact of velocity of users on NMSE of DL channel prediction with $Q_{\rm{s}}=N_{\rm r}$, $N_{\rm f}=2$ and $\rm{SNR}=15$ dB.\vspace{-0pt}}
\label{fig10}
\end{figure}
\vspace{-0pt}
\begin{figure}[htbp]
\centerline{\includegraphics[width=0.46\textwidth]{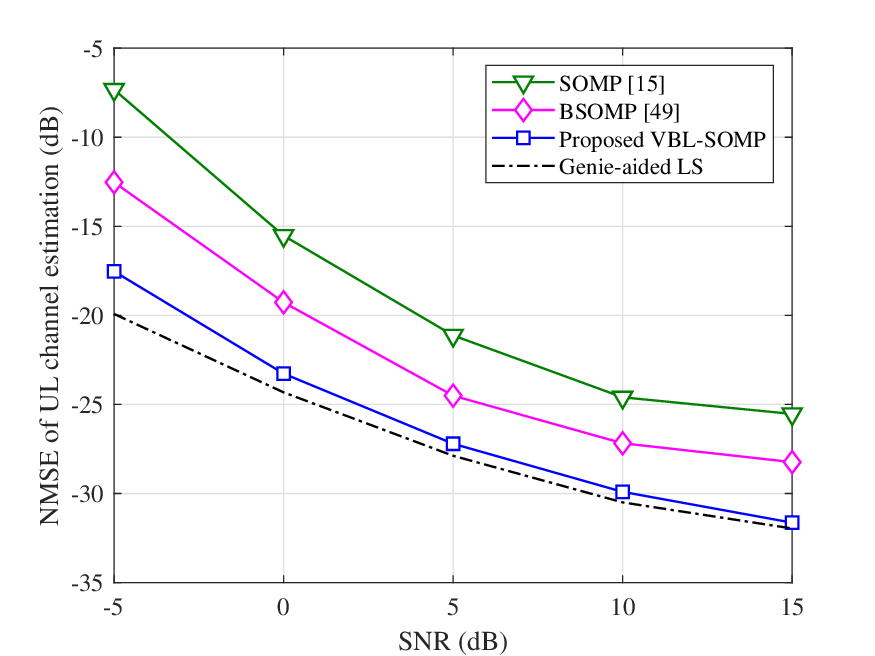}}
\caption{NMSE of UL channel estimation vs. SNR with $Q_{\rm s}=N_{\rm r}$ and $v=120$ km/h.\vspace{-0pt}}
\label{fig11}
\end{figure}

Fig. 9 investigates the NMSE performance of DL channel prediction of the proposed SBEE predictor under different lengths of predicted OTFS frames and different SNRs. The user's velocity and the order of SR-BEM are set to \mbox{$v= 120$ km/h} and $Q_{\rm s}=32$, respectively. From Figs. 9a) and b), it can be seen that the proposed SBEE channel predictor outperforms the ST-AR \cite{Wu2021}, Kalman filter \cite{Kalman2000}, vector Prony \cite{Prony2020} and LSTM predictors \cite{Mattu2022} in terms of NMSE of DL channel prediction, especially at low SNR regime. For example, at SNR = -5 dB and $N_{\rm f}=2$, the proposed SBEE channel predictor achieves about 13 dB gain in terms of NMSE compared to that of the existing vector Prony scheme \cite{Prony2020}. Moreover, the NMSE performance of the UL channel samples aided SBEE channel predictor is very close to that of the perfect UL channel aided one. This is mainly attributed to the accurate BEM iterative UL channel estimation. As $N_{\rm f}$ increases, the iterative process of SBEE predictor amplifies error propagation from UL channel estimation, leading to a widening performance gap compared to SBEE with perfect UL channel. However, as SNR increases, UL channel estimation becomes more accurate, reducing error propagation and narrowing the channel prediction performance gap. The performance of Kalman filter-based channel predictor \cite{Kalman2000} degrades when UL channel estimates rather than true UL channel is used to obtain the temporal autocorrelation function. Hence, it exhibits worse performance than the proposed scheme in Figs. 8 and 9.

\begin{figure}[htbp]
\centerline{\includegraphics[width=0.46\textwidth]{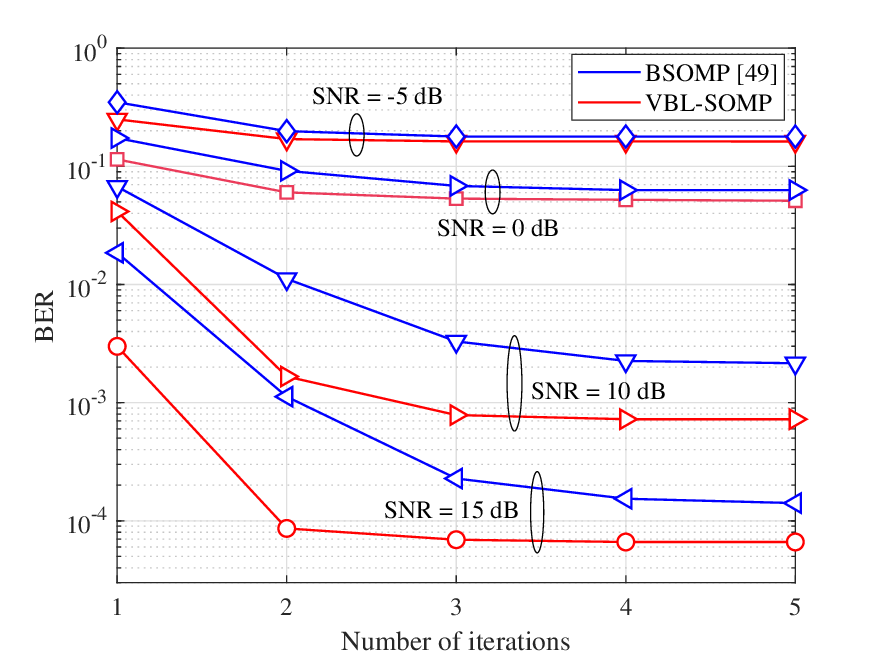}}
\caption{Convergence speed of the proposed channel estimator and existing BSOMP scheme \cite{Gong2017} with $Q_{\rm s}=N_{\rm r}$ and $v=120$ km/h.}
\vspace{-0pt}
\label{fig12}
\end{figure}
\begin{figure}[htbp]
\centerline{\includegraphics[width=0.46\textwidth]{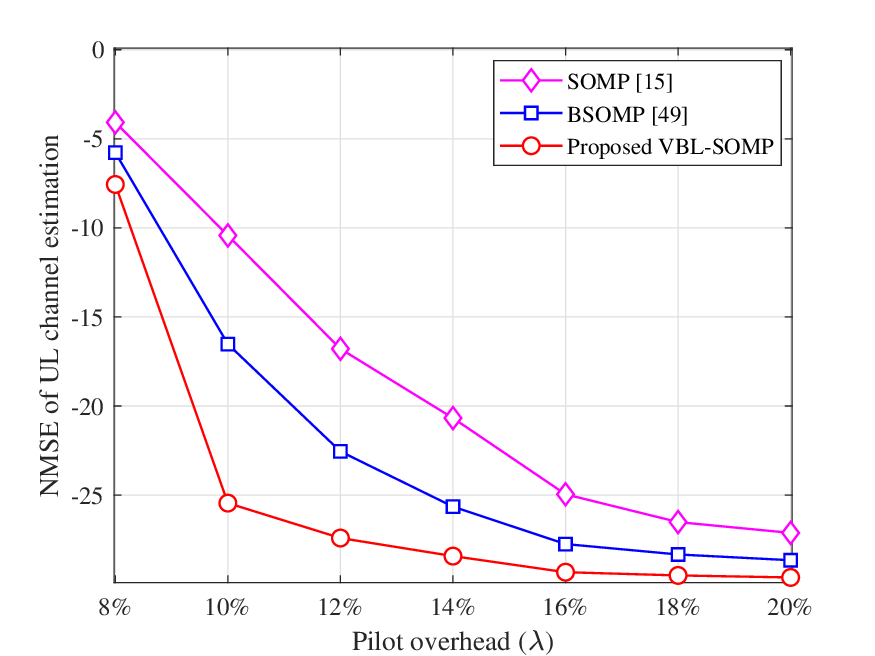}}
\caption{NMSE of UL channel estimation vs. pilot overhead with $Q_{\rm s}=N_{\rm r}$, $v= 120$ km/h and $\rm{SNR}=10$ dB.\vspace{-0pt}}
\label{fig13}
\end{figure}

Fig. 10 investigates the impact of velocity of users on NMSE performance of DL channel prediction for $N=$ 1, 4, 8 and 16. Note that when $N = 1$, the OTFS system turns out to be OFDM system. As the user speed increases, both the Prony predictor \cite{Prony2020} and the proposed SBEE predictor experience degraded channel prediction performance. This degradation stems from two factors: higher user speeds introduce greater modeling errors in the BEM representation of the UL channel, resulting in larger UL channel estimation errors; and increased mobility reduces channel temporal correlation, thereby degrading the performance of channel predictors. As $N$ increases, the NMSE performance of SBEE channel predictor and existing vector Prony predictor \cite{Prony2020} decreases. This is mainly because the larger the value of $N$, the longer the duration of OTFS frame, resulting in more channel coefficients to be predicted. This implies that channel prediction in OTFS systems is more challenging than that in OFDM systems. Note that the degraded DL channel prediction performance can be alleviated by enhancing UL channel estimation, \emph{e.g.}, increasing SNR, using more \mbox{pilots, etc.}

\begin{figure}[htbp]
\centerline{\includegraphics[width=0.46\textwidth]{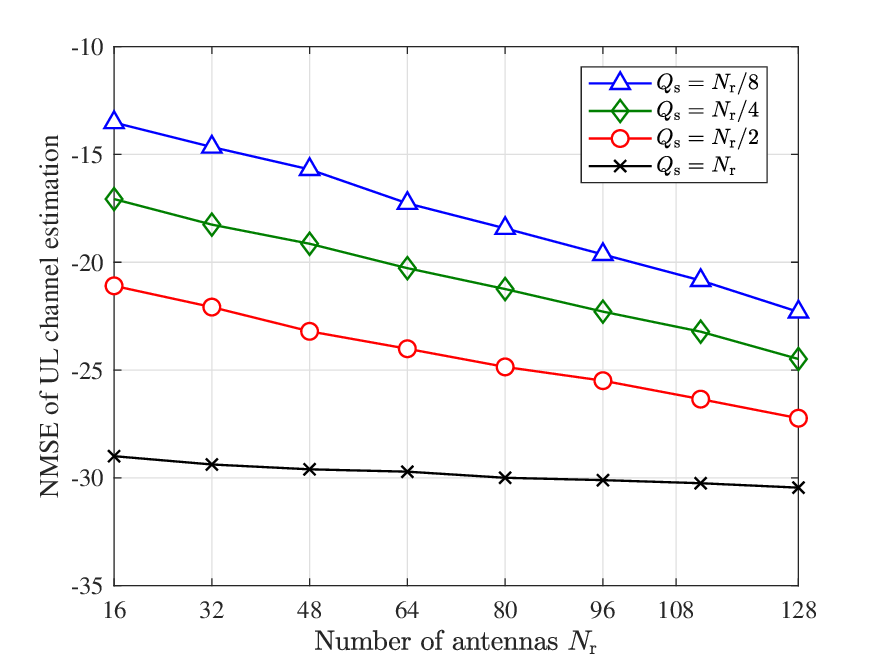}}
\caption{NMSE of UL channel estimation vs. number of BS antennas with $v= 120$ km/h and $\rm{SNR}=10$ dB.\vspace{-10pt}}
\label{fig14}
\end{figure}

\vspace{-0pt}
\subsection{UL Channel Estimation Performance}
In this subsection, the performance of the proposed BEM UL channel estimation scheme is verified by comparing with existing schemes in terms of convergence, pilot overhead, and NMSE performance.

Fig. 11 shows the UL channel estimation performance of the proposed BEM channel estimator, in comparison to the existing SOMP \cite{Shen2019} and BSOMP \cite{Gong2017} schemes. It can be seen that the proposed VBL-SOMP algorithm outperforms the existing SOMP \cite{Shen2019} and BSOMP \cite{Gong2017} bsaed schemes in terms of NMSE of channel estimation. In addition, the proposed VBL-SOMP algorithm approaches the NMSE achieved by the genie-aided LS algorithm. Unlike the existing \cite{Shen2019} and BSOMP \cite{Gong2017} based schemes, the VBL-SOMP algorithm can exploit the common block sparsity structure, thus enhancing the estimation accuracy of non-zero channel paths.


Fig. 12 investigates the convergence speed of the proposed BEM UL channel estimator in terms of BER. At low SNRs, the BERs corresponding to the BSOMP \cite{Gong2017} and the proposed VBL-SOMP channel estimation algorithm remains almost unchanged after 2 iterations. At higher SNRs, the BER corresponding to the proposed VBL-SOMP channel estimation algorithm remains constant after 3 iterations, which indicates that the proposed BEM UL channel estimator converges fast.


Fig. 13 presents the NMSE of channel estimation versus the pilot overhead for the proposed BEM UL channel estimation scheme and the comparators. The NMSE of the proposed BEM UL channel estimation scheme remains unchanged when the pilot overhead larger than 16\%, which means that the proposed scheme only needs 16\% pilot overhead to achieve accurate channel estimation. At ${\lambda _{\rm{p}}}=16\%$, the proposed BEM UL channel estimation scheme outperforms the existing SOMP \cite{Shen2019} and BSOMP \cite{Gong2017} based schemes in terms of NMSE of channel estimation. This indicates that the proposed BEM UL channel estimation scheme is able to provide good historical UL channel samples for subsequent DL channel prediction with very low pilot overhead.

Fig. 14 exhibits the NMSE of channel estimation of the proposed BEM UL channel estimator as a function of the number of BS antennas for $Q_{\rm s}=N_{\rm r}/8$, $Q_{\rm s}=N_{\rm r}/4$, $Q_{\rm s}=N_{\rm r}/2$ and $Q_{\rm s}=N_{\rm r}$. On the one hand, with the decrease of $Q_{\rm s}$, the performance of channel estimation becomes worse. On the other hand, as the number of BS antennas increases, the channel estimation accuracy improves. This is because the spatial resolution of the channel improves as the number of BS antennas increases, leading to more accurate BEM channel modeling and higher channel estimation accuracy.

\vspace{-0pt}
\section{Conclusions}
In this paper, an integrated scheme of UL channel estimation and DL channel prediction has been proposed for TDD massive MIMO-OTFS systems. The proposed iterative BEM channel estimator achieves accurate UL channel estimation with low pilot overhead and fast convergence speed. The proposed SBEE channel predictor enables accurate and long-term DL channel prediction by using a small number of estimated UL channel samples. The proposed SBEE channel predictor outperforms the existing AR \cite{Wu2021}, Kalman filter \cite{Kalman2000}, vector Prony \cite{Prony2020} and LSTM \cite{Mattu2022} channel predictors in terms of NMSE of DL channel prediction, and provide a close DL SE performance to the upper bound obtained by assuming perfect DL channel. Future work would study the extension of the proposed channel estimation and channel prediction schemes in the Terahertz band when quantization error, phase noise and inter-cell interference are considered.

\begin{appendices}
\vspace{-0pt}
\section{Derivation of (13)}

By substituting (12) and (2) into (1), the total received signal in terms of SR-BEM coefficients is given as 

\begin{equation}
{{\bf{y}}_{{\rm{DD}}}} = \sum\nolimits_{q = 0}^{Q - 1} {{{{\bf{\bar U}}}_q}{\bf{\Psi \bar D}}{{\bf{s}}_q}}  + {\bf{z}}.
\label{eq41}
\end{equation}

By multiplying both the left and right sides of (41) by ${\bf{\bar U}}_{\frac{{Q - 1}}{2}}^{\rm H}$, (41) can be expressed as 
\begin{equation}
\begin{aligned}
{\bf{\bar U}}_{\frac{{Q - 1}}{2}}^{\rm H}{{\bf{y}}_{{\rm{DD}}}} &= \sum\nolimits_{q' = 0}^{Q - 1} {{\bf{\bar U}}_{\frac{{Q - 1}}{2}}^{\rm H}{{{\bf{\bar U}}}_{q'}}{\bf{\Psi \bar D}}{{\bf{s}}_{q'}} + {\bf{z}}} \\
&\mathop  = \limits^{(b)} \sum\nolimits_{q' = 0}^{Q - 1} {({{\bf{I}}_{{N_{\rm r}}}} \otimes {{\bf{\Pi }}^{(q' - \frac{{Q - 1}}{2})}}){\bf{\Psi \bar D}}{{\bf{s}}_{q'}} + {\bf{z}}} \\
&\mathop  = \limits^{(c)} {\bf{\Psi \bar D}}{{\bf{s}}_q} + {\bf{\bar z}},
\end{aligned}
\label{eq42}
\end{equation}
where ${\bf{\bar z}} = \sum\limits_{q = 0,q \ne \frac{{Q - 1}}{2}}^{Q - 1} {({{\bf{I}}_{{N_{\rm r}}}} \otimes {{\bf{\Pi }}^{(q' - \frac{{Q - 1}}{2})}})} {\bf{\Psi \bar D}}{{\bf{s}}_q} + {\bf{z}}$, (b) holds because the following formula holds
\begin{equation}
\begin{aligned}
{\bf{\bar U}}_{\frac{{Q - 1}}{2}}^{\rm H}{{{\bf{\bar U}}}_{q'}} &= {({{\bf{I}}_{{N_{\rm r}}}} \otimes {{\bf{U}}_{\frac{{Q - 1}}{2}}})^{\rm H}}\left( {{{\bf{I}}_{{N_{\rm r}}}} \otimes \left( {{{\bf{U}}_{q'}}} \right)} \right)\\
 &= {{\bf{I}}_{{N_{\rm r}}}} \otimes ({{\bf{F}}_{MN}}{\mathop{\rm diag}\nolimits} ({\bf{b}}_{\frac{{Q - 1}}{2}}^*){{\bf{B}}_{\rm{t}}}{{\bf{B}}_{\rm{r}}}{\mathop{\rm diag}\nolimits} ({{\bf{b}}_{q'}}){\bf{F}}_{MN}^{\rm H})\\
 &= {{\bf{I}}_{{N_{\rm r}}}} \otimes {{\bf{\Pi }}^{(q' - \frac{{Q - 1}}{2})}},
\end{aligned}
\label{eq43}
\end{equation}
and (c) holds because ${{\bf{I}}_{{N_{\rm r}}}} \otimes {{\bf{\Pi }}^{(q' - \frac{{Q - 1}}{2})}} = {{\bf{I}}_{{N_{\rm r}}MN}}$ when ${q' = \frac{{Q - 1}}{2}}$. Define the sampling matrix corresponding to the $q$-th CE-BEM coefficient vector as ${{\bf{\Omega }}_q} = {{\bf{I}}_{{N_{\rm r}}}} \otimes {[{{\bf{I}}_{MN}}]_{{{\cal P}_q}}}$, By sampling equation (42), we can obtain
\begin{equation}
\begin{aligned}
{[{{{\bf{\tilde y}}}_{{\rm{DD}}}}]_{{{\cal P}_q}}} &= \sum\limits_{q' = 0}^{Q - 1} {{{\bf{\Omega }}_q}({{\bf{I}}_{{N_{\rm r}}}} \otimes {{\bf{\Pi }}^{(q' - \frac{{Q - 1}}{2})}})} {\bf{\Psi \bar D}}{{\bf{s}}_{q'}} + {\bf{z}}\\
&\mathop  = \limits^{(d)} {{\bf{\Psi }}_{{{\cal P}_{\frac{{Q - 1}}{2}}}}}{\bf{\bar D}}{{\bf{s}}_{q'}} + {\bf{z}},
\end{aligned}
\label{eq44}
\end{equation}
where ${{{\bf{\tilde y}}}_{{\rm{DD}}}} = {\bf{\bar U}}_{\frac{{Q - 1}}{2}}^{\rm H}{{\bf{y}}_{{\rm{DD}}}}$, ${[{{{\bf{\tilde y}}}_{{\rm{DD}}}}]_{{{\cal P}_q}}} = {{\bf{\Omega }}_q}{{{\bf{\tilde y}}}_{{\rm{DD}}}}$, and (d) holds because
\begin{equation}
{{\bf{\Omega }}_q}({{\bf{I}}_{{N_{\rm r}}}} \otimes {{\bf{\Pi }}^{(q' - \frac{{Q - 1}}{2})}}) = \left\{ {\begin{array}{*{20}{c}}
{{{\bf{\Omega }}_{{{\cal P}_{\frac{{Q - 1}}{2}}}}},{\rm{  }}q = q'}\\
{{{\bf{\Omega }}_{q - q' - \frac{{Q - 1}}{2}}},{\rm{  }}q \ne q'}
\end{array}} \right.
\label{eq45}
\end{equation}
Note that for $q \ne \frac{{Q - 1}}{2}$, ${{\bf{\Psi }}_{{{\cal P}_q}}} = {{\bf{0}}_{{N_{\rm r}}G}} \otimes {{\bf{1}}_{{N_{\rm u}}}}$. The $Q$ separated equations in (13) can be easily obtained from (44). 

\vspace{-0pt}

\section{Proof of Lemma 1}
To simplify the derivation, the modeling error using SR-BEM in equation (11) is equal to 0 by assuming $Q_{\rm s}=N_{\rm r}$.
Since the UL channels between all antenna pairs are all modeled with the CE-BEM, we omit the antenna index $n_{\rm r}$ and user index $n_{\rm u}$ to simplify the description. Define the non-zero UL channel between any antenna pair as $\{ {\bm{\hbar}}_k\} _{k = 1}^{K}$, according to (8), the modeling error for the $k$-th path can be expressed as:
\begin{equation}
{{\bf{v}}_k} = {{\bm{\hbar}}_k} - {\bf{B}}{{\bf{c}}_k} = ({{\bf{I}}_{MN}} - {\bf{B}}{{\bf{B}}^{H} }){{\bm{\hbar}}_k}.
\label{eq46}
\end{equation}

Note that ${{\bf{B}}^{H} }{\bf{B}} = {{\bf{I}}_{{Q_{{\rm{}}}}}}$. Consider all antenna pairs, the average error of UL channel model is given by
\begin{equation}
\begin{array}{l}
{\rm{MS}}{{\rm{E}}_{\bmod }} = \frac{1}{{\bar \omega }}\mathbb{E}\{ \bm{\hbar} _k^{\rm H}({{\bf{I}}_{MN}} - {\bf{B}}{{\bf{B}}^{\rm H}}){\bm{\hbar} _k}\} \\
 = \frac{1}{{\bar \omega }}{\rm{Tr}}\left[ {({{\bf{I}}_{MN}} - {\bf{B}}{{\bf{B}}^{\rm H}}){\bf{\bar R}}} \right],
\end{array}
\label{eq47}
\end{equation}
where $\bar \omega  = \frac{1}{{MNK{N_{\rm r}}{N_{\rm u}}{N_{\rm t}}}}$, $\bf{\bar R}$ is the sum of the autocorrelation matrices of the channels of all paths, whose $t$-th row and $ t'$-th column is given as ${[{\bf{\bar R}}]_{t,{t^\prime }}} = {J_0}\left( {2\pi {f_{\rm d}}{T_{\rm s}}\left( {t - {t^\prime }} \right)} \right)$ \cite{Zhang2022}.
As $M \to \infty $, the bandwidth tends to infinity, and the sampling period ${T_{\rm s}} \to 0$ and thus ${\bf{\bar R}} \to {{\bf{1}}_{MN \times MN}}$, one can obtain
\begin{equation}
\begin{array}{l}
\mathop {\lim }\limits_{M \to \infty } {\rm{MS}}{{\rm{E}}_{\bmod }} = \\
\frac{1}{{\bar \omega }}\mathop {\lim }\limits_{M \to \infty } (MN - \sum\limits_{i = 0}^{MN - 1} {\sum\limits_{j = 0}^{MN - 1} {{{{\bf{\bar b}}}_i}{\bf{\bar b}}_j^{\rm H}} } ) = 0,
\end{array}
\label{eq48}
\end{equation}
where ${{{{\bf{\bar b}}}_i}}$ is the $i$-th row of $\bf B$ and $\sum\nolimits_{i = 1}^{MN} {\sum\nolimits_{j = 1}^{MN} {{{{\bf{\bar b}}}_i}{\bf{\bar b}}_j^{\rm H}} }  = MN$ because $\bf B$ is a truncated orthogonal matrix.

Define ${\bar {\bm{\hbar}}}  = {[{({\bf{{ h}}}_{0,1,1}^{{\rm{UL}}})^{\rm H}}, \cdots ,{({\bf{{ h}}}_{{N_{\rm r}} - 1,{N_{\rm u}},K}^{{\rm{UL}}})^{\rm H}}]^{\rm H}} $ as the channel of any UL OTFS frame. Neglecting BEM modeling errors, ${\bar {\bm{\hbar}}}$ can be expressed as a function of SR-BEM coefficients ${\bar {\bm{\hbar}}} \buildrel \Delta \over = {\widetilde {\bf D}}{\bar {\bf{s}}}$,
where ${\bar {\bf{s}}} = {\mathop{\rm vec}\nolimits} ({\bf{S}})$ and ${\widetilde {\bf D}} = ({{\bf{I}}_{{N_{\rm u}}}} \otimes ({{\bf{I}}_{{N_{\rm r}}}} \otimes ({{\bf{I}}_K} \otimes {\bf{B}})))({{\bf{I}}_Q} \otimes {\bf{\bar D}}) \in {\mathbb{C}^{MN{N_{\rm r}}{N_{\rm u}}K \times {N_{\rm u}}KQ{N_{\rm r}}}}$. Note that ${{{\bf{\tilde D}}}^{\rm H}}{\bf{\tilde D}} = {{\bf{{ I}}}_{{N_{\rm r}}{N_{\rm u}}KQ}}$. The average MSE for estimating the channels of $N_{\rm t}$ UL OTFS frames can be calculated as
\begin{equation}
\begin{aligned}
{\rm{MS}}{{\rm{E}}_{{\rm{est}}}} &= \bar \omega \sum\nolimits_{{n_{\rm t}} = 1}^{{N_{\rm t}}} {\mathbb{E}\{ ||{\bm{\bar{\hbar}}} - {\bm{\hat{\bar{\hbar}}}}|{|^2}\} } \\
 &= \bar \omega \sum\nolimits_{{n_{\rm t}} = 1}^{{N_{\rm t}}} {\mathbb{E}\{ {{({\bf{\bar s}} - {\bf{\hat {\bar s}}})}^{\rm H}}{{{\bf{\tilde D}}}^{\rm H}}{\bf{\tilde D}}({\bf{\bar s}} - {\bf{\hat {\bar s}}})\} } \\
 &= \bar \omega \sum\nolimits_{{n_{\rm t}} = 1}^{{N_{\rm t}}} {\mathbb{E}\{ ||{\bf{\bar s}} - {\bf{\hat {\bar s}}}|{|^2}\} } ,
\end{aligned}
\label{eq49}
\end{equation}
where ${{\bf{\hat {\bar s}}}}$ is the estimated SR-BEM coefficient vector. As shown in \cite{Liu2022}, according to the vector estimation theory, we can obtain ${\mathbb E}\{ ||{\bf{\bar s}} - {\bf{\hat {\bar s}}}|{|^2}\}  = {\rm{Tr}}\{ {({{\bf{\Phi }}^{\rm{H}}}{\bf{\Phi }})^{ - 1}}\} $, where $\bf \Phi$ is the measurement matrix defined in (14). Given the number of antennas $N_{\rm r}$, the number of users $N_{\rm u}$ and the number of non-zero pilots $G$, ${\rm{Tr}}\{ {({{\bf{\Phi }}^{\rm{H}}}{\bf{\Phi }})^{ - 1}}\} $ is a constant, and the estimation error depends only on the noise power ${\sigma}^2$. The asymptotic MSE of channel estimation can be expressed as
 \begin{equation}
\mathop {\lim }\limits_{{\sigma ^2} \to 0} {\rm{MS}}{{\rm{E}}_{{\rm{est}}}} = \mathop {\lim }\limits_{{\sigma ^2} \to 0}\bar \omega\sum\nolimits_{{n_{\rm t}} = 1}^{{N_{\rm t}}} {\mathbb E}\{ ||{\bar {\bm{\hbar}}} -\hat {\bar {\bm{\hbar}}}|{|^2}\}  = 0.
 \label{eq50}
\end{equation}

According to CS theory \cite{Tropp2007}, $K$-sparse signals can be accurately recovered from $G > K{\rm{log}}(L{\rm{/}}K)$ linear measurements by using OMP-like algorithms. This reveals that accurate UL channel estimation can be achieved by using the proposed VBL-SOMP algorithm when the number of pilots is large enough. As shown in (47) and (49), as ${\sigma ^2} \to 0$ and $M \to \infty $, one can obtain error-free UL channels. 

\vspace{-10pt}
\section{Proof of Theorem 1}
Note that (23) can be vectorized as 
\begin{equation}
{{{\bf{\hat h}}}^{{\rm{UL}}}} = ({{\bf{I}}_{{N_{\rm r}}{N_{\rm u}}K}} \otimes {{{\bf{\bar B}}}_{{\rm{SP}}}}){{\bf{c}}_{{\rm{SP}}}^{\rm UL}} + {{\bf{w}}_{{\rm{SP}}}^{\rm UL}},
\label{eq51}
\end{equation}
where ${{{\bf{\hat h}}}^{{\rm{UL}}}} = {\mathop{\rm vec}\nolimits} ({{{\bf{\hat H}}}^{{\rm{UL}}}})$, ${{\bf{c}}_{{\rm{SP}}}^{\rm UL}} = {\mathop{\rm vec}\nolimits} ({{\bf{C}}_{{\rm{SP}}}^{\rm UL}})$ and ${{\bf{w}}_{{\rm{SP}}}^{\rm UL}} = {\mathop{\rm vec}\nolimits} ({{\bf{W}}_{{\rm{SP}}}^{\rm UL}})$. The asymptotic modeling error of ${{{\bf{\hat h}}}^{{\rm{UL}}}}$ can be expressed as
\begin{equation}
\begin{array}{l}
\mathop {\lim }\limits_{{\sigma ^2} \to 0,M \to \infty } \bar \omega \mathbb{E}\{ {({\bf{w}}_{{\rm{SP}}}^{{\rm{UL}}})^{\rm H}}{\bf{w}}_{{\rm{SP}}}^{{\rm{UL}}}\} \\
 = \mathop {\lim }\limits_{{\sigma ^2} \to 0,M \to \infty } \bar \omega \mathbb{E} \{ ||({{\bf{I}}_{MN{N_{\rm r}}{N_{\rm u}}K{N_{\rm t}}}} - {{{\bf{\tilde B}}}_{{\rm{SP}}}}{\bf{\tilde B}}_{{\rm{SP}}}^{\rm T}){{{\bf{\hat h}}}^{{\rm{UL}}}}|{|^2}\} \\
 = \mathop {\lim }\limits_{{\sigma ^2} \to 0,M \to \infty } (1 - \bar \omega {N_{\rm r}}{N_{\rm u}}K{N_{\rm t}}{\rm{Tr}}[{{\bf{B}}_{{\rm{SP}}}}{\bf{B}}_{{\rm{SP}}}^{\rm T}{\bf{\bar R}}]) = 0,
\end{array}
\label{eq52}
\end{equation}
where ${{{\bf{\tilde B}}}_{{\rm{SP}}}} = {{\bf{I}}_{{N_{\rm t}}{N_{\rm r}}{N_{\rm u}}K}} \otimes {{\bf{B}}_{{\rm{SP}}}}$. The $N_{\rm f}$-frame DL channel can be modeled by Slepian sequences as ${{\bf{H}}^{{\rm{DL}},{N_{\rm f}}}} = {{{\bf{\bar B}}}_{{\rm{SP}}}}{\bf{C}}_{{\rm{DPS}}}^{{\rm{DL}}} + {\bf{W}}_{{\rm{SP}}}^{{\rm{DL}}}$, where ${{\bf{H}}^{{\rm{DL}},{N_{\rm f}}}}$ is the true DL channel, ${\bf{C}}_{{\rm{SP}}}^{{\rm{DL}}}$ is the DL Slepian coefficient matrix and ${\bf{W}}_{{\rm{SP}}}^{{\rm{DL}}}$ is the DL Slepian modeling error. Similar to the analysis in (51), the asymptotic error of DL channel modeling is given as $\mathop {\lim }\limits_{M \to \infty } \mathbb{E}\{ {({\bf{w}}_{{\rm{SP}}}^{{\rm{DL}}})^{\rm H}}{\bf{w}}_{{\rm{SP}}}^{{\rm{DL}}}\}  = 0$, where ${\bf{w}}_{{\rm{SP}}}^{{\rm{DL}}} = {\mathop{\rm vec}\nolimits} ({\bf{W}}_{{\rm{SP}}}^{{\rm{DL}}})$. Since $\mathop {\lim }\limits_{M \to \infty } {({\bf{h}}_{{n_{\rm r}},{n_{\rm u}},k}^{{\rm{UL,}}{n_{\rm t}}})^{\rm H}}{\bf{h}}_{{n_{\rm r}},{n_{\rm u}},k}^{{\rm{UL,}}{n_{\rm t}}} = \mathop {\lim }\limits_{M \to \infty } {({\bf{h}}_{{n_{\rm r}},{n_{\rm u}},k}^{{\rm{DL,}}{n_{\rm f}}})^{\rm H}}{\bf{h}}_{{n_{\rm r}},{n_{\rm u}},k}^{{\rm{DL,}}{n_{\rm f}}} = {{\bf{1}}_{MN \times MN}}$, $\forall $ $n_{\rm t}=1,\cdots,N_{\rm t}$ and $n_{\rm f}=1,\cdots,N_{\rm f}$, where ${\bf{h}}_{{n_{\rm r}},{n_{\rm u}},k}^{{\rm{UL,}}{n_{\rm t}}}$ and ${\bf{h}}_{{n_{\rm r}},{n_{\rm u}},k}^{{\rm{DL,}}{n_{\rm t}}}$ are the true UL and DL channel, respectively, we have $\mathop {\lim }\limits_{{\sigma ^2} \to 0,M \to \infty } ||{\bf{c}}_{{\rm{SP}}}^{{\rm{UL}}} - {\bf{c}}_{{\rm{SP}}}^{{\rm{DL}}}|{|^2} = 0$, where ${\bf{c}}_{{\rm{SP}}}^{{\rm{DL}}} = {\mathop{\rm vec}\nolimits} ({\bf{C}}_{{\rm{SP}}}^{{\rm{DL}}})$. In this case, only a zero-order DLP is sufficient to accurately fit the Slepian coefficients. Assuming that the number of iterations of SBEE is $p=1$, and the iteration step size is $\Delta=N_{\rm f}$, the $n_{\rm t}$-th Slepian coefficient can be obtained by
\begin{equation}
 \bar C_{{\rm{SP}}}^{}({n_{\rm t}},k) = \hat c_{k,0}^{{\rm{DLP}},(0)}\varphi _0^{(0)}[{t_{{n_{\rm t}}}}] = \hat c_{k,0}^{{\rm{DLP}},(0)},
\label{eq53}
\end{equation}
and the $(N_{\rm t}+N_{\rm f})$-th Slepian coefficient can be obtained by
\begin{equation}
 \hat {\bar C}_{{\rm{SP}}}^{(1)}({N_{\rm t}} + {n_{\rm f}},k) = \hat c_{k,0}^{{\rm{SP}},(0)}\varphi _0^{(1)}[{t_{{N_{\rm t}} + {n_{\rm f}}}}] = \hat c_{k,0}^{{\rm{DLP}},(0)},
\label{eq54}
\end{equation}
where $\varphi _0^{(0)}[{t_{{n_{\rm t}}}}] = \varphi _0^{(1)}[{t_{{N_{\rm t}} + {n_{\rm f}}}}] = 1$. From (52) and (53), we have $\mathop {\lim }\limits_{{\sigma ^2} \to 0,M \to \infty } ||{\bf{c}}_{{\rm{SP}}}^{{\rm{UL}}} - {\bf{\hat c}}_{{\rm{SP}}}^{{\rm{DL}}}|| = 0$, and thus $\mathop {\lim }\limits_{{\sigma ^2} \to 0,M \to \infty } ||{\bf{c}}_{{\rm{SP}}}^{{\rm{DL}}} - {\bf{\hat c}}_{{\rm{SP}}}^{{\rm{DL}}}|| = 0$. Finally, we have 
\begin{equation}
\begin{array}{l}
\mathop {\lim }\limits_{{\sigma ^2} \to 0,M \to \infty }\mathbb{E} \{ ||{{\bf{h}}^{{\rm{DL}},{N_{\rm f}}}} - {{{\bf{\hat h}}}^{{\rm{DL}},{N_{\rm f}}}}|{|^2}\} \\
 = \mathop {\lim }\limits_{{\sigma ^2} \to 0,M \to \infty }\mathbb{E} \{ ||{{{\bf{\tilde B}}}_{{\rm{SP}}}}({\bf{c}}_{{\rm{SP}}}^{{\rm{DL}}} - {\bf{\hat c}}_{{\rm{SP}}}^{{\rm{DL}}})|{|^2}\} \\
 = \mathop {\lim }\limits_{{\sigma ^2} \to 0,M \to \infty }\mathbb{E} \{ {({\bf{c}}_{{\rm{SP}}}^{{\rm{DL}}} - {\bf{\hat c}}_{{\rm{SP}}}^{{\rm{DL}}})^{\rm H}}{\bf{\tilde B}}_{{\rm{SP}}}^{\rm T}{{{\bf{\tilde B}}}_{{\rm{SP}}}}({\bf{c}}_{{\rm{SP}}}^{{\rm{DL}}} - {\bf{\hat c}}_{{\rm{SP}}}^{{\rm{DL}}})\\
 = \mathop {\lim }\limits_{{\sigma ^2} \to 0,M \to \infty }\mathbb{E} \{ ||{\bf{c}}_{{\rm{SP}}}^{{\rm{DL}}} - {\bf{\hat c}}_{{\rm{SP}}}^{{\rm{DL}}}|{|^2}\}  = 0,
\end{array}
\label{eq55}
\end{equation}
where ${{{\bf{\tilde B}}}_{{\rm{SP}}}} = {{\bf{I}}_{{N_{\rm r}}{N_{\rm u}}K{N_{\rm f}}}} \otimes {{\bf{B}}_{{\rm{SP}}}}$, ${{\bf{h}}^{{\rm{DL}},{N_{\rm f}}}} = {\mathop{\rm vec}\nolimits} ({{\bf{H}}^{{\rm{DL}},{N_{\rm f}}}})$ and ${{{\bf{\hat h}}}^{{\rm{DL}},{N_{\rm f}}}} = {\mathop{\rm vec}\nolimits} ({{{\bf{\hat H}}}^{{\rm{DL}},{N_{\rm f}}}})$.

\end{appendices}

\vspace{-0pt}

\bibliographystyle{IEEEbib}
\bibliography{refs}

\end{document}